\documentclass[singlecolumn,floatfix]{revtex4-1}
\usepackage{amsfonts}
\usepackage{amsmath}
\usepackage{amssymb}
\usepackage{amsthm}
\usepackage{array}
\usepackage{epsfig}
\usepackage{rotating,graphicx}
\usepackage{stix}
\usepackage{bm}%
\usepackage{verbatim}
\usepackage{tensor}
\usepackage{tikz}
\usepackage{longtable}
\usepackage{textgreek}
\theoremstyle{plain}
\newtheorem{theorem}{Theorem}
\newtheorem{corollary}[theorem]{Corollary}

\newtheorem{lemma}{Lemma}

\theoremstyle{definition}
\newtheorem{definition}{Definition}

\newtheorem{exatitle}{Example}
\newenvironment{myexample}[2]%
{\begin{exatitle} \label{#2} #1 \end{exatitle}}%
{\hfill $\Box$ \\}

\newcommand{\ket}[1]{| #1 \rangle}
\newcommand{\bra}[1]{\langle #1 |}
\newcommand{\braket}[2]{\langle #1 | #2 \rangle}
\newcommand{\pket}[1]{[ #1 ]}

\newcommand{\td}{\text{d}}
\newcommand{\Tr}{\text{Tr}}

\newcommand{\eg}{\hbox{\em e.g.{}}}
\newcommand{\etc}{\hbox{\em etc.{}}}
\newcommand{\ie}{\hbox{\em i.e.{}}}

\newcommand{\rhs}{\hbox{r.h.s.{}}}

\newcommand{\Ps}{\mathbb{P}}

\newcommand{\rha}[1]{\rightharpoonaccent{#1}}
\newcommand{\ip}[2]{\langle #1|#2 \rangle}
\newcommand{\ipc}[2]{\langle #1,#2 \rangle}
\newcommand{\ipp}[2]{\left( #1,#2 \right)}

\makeatletter
\g@addto@macro\bfseries{\boldmath}
\makeatother
\begin{document}

\title{Stellar Representation of Grassmannians}

\author{C.{} Chryssomalakos}
\email{chryss@nucleares.unam.mx}
\affiliation{Instituto de Ciencias Nucleares \\
	Universidad Nacional Aut\'onoma de M\'exico\\
	PO Box 70-543, 04510, CDMX, M\'exico.}

\author{E.{} Guzm\'an-Gonz\'alez}
\email{egomoshi@gmail.com}
\affiliation{Instituto de Ciencias Nucleares \\
	Universidad Nacional Aut\'onoma de M\'exico\\
	PO Box 70-543, 04510, CDMX, M\'exico.}

\author{L.{} Hanotel}
\email{hanotel@correo.nucleares.unam.mx}	
\affiliation{Instituto de Ciencias Nucleares \\
	Universidad Nacional Aut\'onoma de M\'exico\\
	PO Box 70-543, 04510, CDMX, M\'exico.}

\author{E.{} Serrano-Ens\'astiga}
\email{eduardo.serrano-ensastiga@uni-tuebingen.de}
\affiliation{Institut f\"ur Theoretische Physik \\
	Universit\"at T\"ubingen\\
	72076, T\"ubingen, Germany}

\begin{abstract}
	\noindent 
	Pure quantum spin-$s$ states can be represented by $2s$ points on the sphere, as shown by Majorana in 1932 --- the description has proven particularly useful in the study of rotational symmetries of the states, and a host of other properties, as the points rotate rigidly on the sphere when the state undergoes an $SU(2)$ transformation in Hilbert space. At the same time, the Wilzcek-Zee effect, which involves the cyclic evolution of a degenerate $k$-dimensional linear subspace of the Hilbert space, and the associated holonomy dictated by Schroedinger's equation,  have been proposed as a fault-tolerant mechanism for the implementation of logical gates, with applications in quantum computing. We show, in this paper, how to characterize such subspaces by Majorana-like sets of points on the sphere, that also rotate rigidly under $SU(2)$ transformations --- the construction is actually valid for arbitrary totally antisymmetric $k$-partite qudit states. 
\end{abstract}

\maketitle

\tableofcontents
\section{Introduction}
\label{Intro}
A quantum spin-$s$ state $\ket{\psi}$ is represented by a ray, \ie, a 1-dimensional linear subspace, in a $\tilde{N}$-dimensional Hilbert space $\mathcal{H}$ ($\tilde{N} \equiv 2s+1$), \ie, a point in the corresponding projective space $\Ps =\mathbb{C}P^N$ ($N\equiv 2s$). In a relatively little known 1932 paper~\cite{Maj:32}, Majorana showed how to uniquely characterize  
$\ket{\psi}$  by an unordered set of (possibly coincident) $2s$ directions in space, \ie, $2s$ points (\emph{stars}) on the unit sphere, known as the \emph{Majorana constellation} of $\ket{\psi}$ (see, \eg,~\cite{Ben.Zyc:17}). The construction is such that when $\ket{\psi}$ is transformed in $\mathcal{H}$ by the spin-$s$ irreducible representation of an $SU(2)$ transformation, the associated constellation rotates by the corresponding rotation in physical space. It can be shown that the directions of the Majorana stars  characterize, in the standard way, the states of $2s$ spin-1/2 particles, which, upon complete symmetrization, yield 
$\ket{\psi}$. Even when the spin-$s$ system is not really made up of spin-1/2 particles, the associated directions can be detected experimentally: aligning a Stern-Gerlach apparatus along any of them, the probability of measuring the minimal spin projection, $-s$, is equal to zero~\cite{Chr.Guz.Ser:18}. 

Cyclic evolution of quantum states gives rise to geometric phases, so that, to each closed curve $\gamma$ in $\Ps$, one may associate a phase factor $e^{i\varphi_\gamma}$, which is independent of the time parametrization of $\gamma$~\cite{Ber:84,Aha.Ana:87}. The concept has been generalized to the cyclic evolution of degenerate $k$-dimensional subspaces of $\mathcal{H}$, so that to each closed curve $\gamma$ in the Grassmannian 
$\text{Gr}_{k,\tilde{N}}$ (which is the set of $k$-planes through the origin in $\mathcal{H}$), one may associate a $k \times k$ unitary matrix $U_\gamma$, which, like its abelian analogue above,  does not depend on the time parametrization of $\gamma$~\cite{Wil.Zee:84,Ana:88,Muk.Sim:93}. Both the abelian and non-abelian versions of the effect have been invoked in the realization of quantum gates, their immunity to reparametrizations contributing to the robustness of the resulting quantum computation~\cite{%
Zan.Ras:99,solinas2004robustness,PhysRevLett.102.070502,golovach2010holonomic,solinas2012stability}. These developments have put emphasis on the geometric  concept of a $k$-plane in $\mathcal{H}$, as a natural generalization of that of a ray, which corresponds to $k=1$, and, inevitably, raise the question whether Majorana's visualization of spin-$s$ rays can be extended to spin-$s$ $k$-planes, the latter henceforth referred to as \emph{$(s,k)$-planes}. Apart from its obvious mathematical appeal (at least to the authors),  the question is well-motivated from a practical point of view, as it simplifies considerably the otherwise awkward task of identifying the possible rotational symmetries of an $(s,k)$-plane. Our aim in this work then is to generalize Majorana's stellar representation of spin-$s$ states, living in $\mathbb{P} \approx \text{Gr}_{1,\tilde{N}}$, to the case of $\text{Gr}_{k,\tilde{N}}$.

In section~\ref{MaP} we give some background information regarding the Majorana constellation and Grassmannians. Our solution to the problem stated above comes in two steps: in the first one, taken in section~\ref{kpCvHf}, we define, in close analogy to Majorana's construction, the principal constellation of an $(s,k)$-plane, which, however, is shown to be shared, for $k>1$,  by many different planes. Section~\ref{kpMvtPE} delivers the second step, by introducing the concept of a multiconstellation, which uniquely identifies an $(s,k)$-plane, for almost all such planes --- several examples illustrate the general theory, as well as its limitations. Finally, section~\ref{Epilogue} summarizes the findings, mentions possible extensions, and outlines a number of applications. Anticipating our discussion there, we mention that our solution in the form of a multiconstellation is actually valid for general spin-$s$ $k$-partite antisymmetric states.

\section{Majorana and Pl\"ucker}
\label{MaP}
\subsection{Majorana Constellations}
The reader is no doubt familiar with the fact that a spin-1/2 pure state may be characterized, up to an overall phase, by a point on the Bloch sphere, which gives the spin expectation value (SEV) of the state. The natural question of whether this visually appealing construction may be generalized to a spin-$s$ state was settled by Majorana in a 1932 paper, dealing with the behavior of spins in variable magnetic fields~\cite{Maj:32}. What Majorana pointed out was the fact that points in the projective Hilbert space $\Ps =\mathbb{C}P^N$ of a spin-$s$ system are in one-to-one correspondence with unordered sets of (possibly coincident) $2s$ points on the unit sphere. Details about this construction may be found in the literature (see, \eg, \cite{Ben.Zyc:17}, \cite{Chr.Guz.Ser:18}), we only present here the bare minimum. 

According to~\cite{Maj:32},  to a spin-$s$ state 
\begin{equation}
|\Psi\rangle=\sum_{m=-s}^s c_m |s,m\rangle
\, ,
\label{sstate}
\end{equation} 
where $S_z \ket{s,m}=m\ket{s,m}$,
one may associate its \emph{Majorana polynomial}
$P_{\ket{\Psi}}(\zeta)$,
\begin{equation}
P_{|\Psi\rangle}(\zeta)
=
\sum_{m=-s}^{s}
(-1)^{s-m}\sqrt{\binom{2s}{s-m}}
c_m \,  \zeta^{s+m}
\, ,
\label{pjk}
\end{equation}
where $\zeta$ is an auxiliary complex variable. The $2s$ roots of $P_{\ket{\Psi}}(\zeta)$, counted with multiplicity, in case some of them coincide, may be mapped to the Bloch sphere by stereographic projection from the south pole, giving rise to the \emph{Majorana constellation} of $\ket{\Psi}$.  Note that if the polynomial turns out of a lower degree, \ie, if $c_{m}=0$ for $m=s, s-1, \ldots, s-k+1$, then $\zeta=\infty$ should be considered a root of multiplicity 
$k$, resulting in the appearance of $k$ stars at the south pole of the Bloch sphere. The remarkable property of this construction is that when $\ket{\Psi}$ is transformed in Hilbert space by the matrix $D^{(s)}(g)$, representing the abstract element $g$ of $SU(2)$, its constellation rotates rigidly by (the rotation in $\mathbb{R}^3$ associated to) $g$ on the Bloch sphere. Thus, if $\ket{\Psi}$ has a particular rotational symmetry, in the sense that there exists an element $g_0 \in SU(2)$ such that $D^{(s)}(g_0)\ket{\Psi}=e^{i\alpha_0}\ket{\Psi}$, its constellation is invariant under $g_0$. The recipe given in~(\ref{pjk}) becomes more transparent by noting that
\begin{equation}
\label{ncohexp}
\ket{n}=\frac{1}{(1+\zeta \bar{\zeta})^s} 
\sum_{m=-s}^s
\sqrt{\binom{2s}{s-m}}\,  \zeta^{s-m} \, \ket{s,m}
\, ,
\end{equation}
where $\ket{n}$ denotes the spin coherent state in the direction $n$, the latter being related to $\zeta$ via stereographic projection, \ie, if the polar coordinates of $n$ are $(\theta,\phi)$, then 
$\zeta=\tan\frac{\theta}{2}\, e^{i \phi}$. 
Given the fact that if $\zeta$ is the stereographic projection of $n$, then $-1/\bar{\zeta}$ is that of $-n$, one gets
\begin{align}
\bra{-n}
&=
\frac{(\zeta \bar{\zeta})^s}{(1+\zeta \bar{\zeta})^s}
\sum_{m=-s}^s \sqrt{\binom{2s}{s-m}} \, (-1)^{s-m} \, \zeta^{-s+m} \, \bra{s,m}
\\
 &=
 \frac{(\bar{\zeta}/ \zeta)^s}{(1+\zeta \bar{\zeta})^s}
\sum_{m=-s}^s \sqrt{\binom{2s}{s-m}} \, (-1)^{s-m} \, \zeta^{s+m} \, \bra{s,m}
\end{align}
resulting, finally, in
\begin{equation}
\label{Majoip}
\braket{-n}{\Psi} = \frac{(\bar{\zeta}/\zeta)^s}{(1+\zeta\bar{\zeta})^s}  P_{\ket{\Psi}}(\zeta)
\, .
\end{equation}
Thus, the stars in the constellation of $\ket{\Psi}$ are antipodal to the directions of all coherent states orthogonal to $\ket{\Psi}$. This, in turn, may be traced to the fact that any spin-$s$ state may be obtained by symmetrization of a factorizable $2s$-qubit state --- see, \eg,~\cite{Chr.Guz.Ser:18} for the details. 
\begin{myexample}{A spin-2 constellation}{ex:s2k1}
Consider the spin-2  state $\ket{\psi_{\text{tetra}}}=(1,0,0,\sqrt{2},0)/\sqrt{3}$. The corresponding Majorana polynomial is 
$P_{\ket{\psi_{\text{tetra}}}}(\zeta)=\zeta^4-2\sqrt{2}\zeta$, with roots $(z_1,z_2,z_3,z_4)=(0,\sqrt{2},e^{i2\pi/3}\sqrt{2},e^{i4\pi/3}\sqrt{2})$, which project to the stars 
\begin{equation}
\label{tetrastars}
(n_1,n_2,n_3,n_4)=
\left(
\left( \rule{0ex}{3.8ex} 0,0,1 \right)
\, , \,
\left(\rule{0ex}{3ex}  -\frac{\sqrt{2}}{3},-\sqrt{\frac{2}{3}},-\frac{1}{3} \right)
\, , \,
\left( \rule{0ex}{3ex} -\frac{\sqrt{2}}{3},\sqrt{\frac{2}{3}},-\frac{1}{3} \right)
\, , \,
\left( \rule{0ex}{3ex} \frac{2\sqrt{2}}{3},0,-\frac{1}{3} \right)
\right)
\, ,
\end{equation}
that define the vertices of a regular tetrahedron. We conclude, \eg, that $\ket{\psi_{\text{tetra}}}$ picks up at most a phase when rotated around any of the above $n_i$ by an angle of $2\pi/3$.
\end{myexample}
\subsection{Some tools for Grassmannians}
\label{StfG}
The \emph{Grassmannian} $\text{Gr}_{k,n}$ is the set of $k$-dimensional linear subspaces (\ie, $k$-planes through the origin) in $\mathbb{C}^n$ (see, \eg, Ch.{} 10 of \cite{Sha.Rem:13}, Ch.{} XIV of \cite{Hod.Ped:52}, or Ch.{} 4.1 of~\cite{Sha:13}). Given a $k$-plane $\Pi \subset \mathbb{C}^n$, and a basis (\ie, a non-degenerate $k$-frame) $\{v_1,\ldots,v_k\}$ in $\Pi$, one may write down the $k \times n$ matrix $V$ of components of the $v$'s,
\begin{equation}
\label{Vdef}
V=
\left(
\begin{array}{ccc}
v_1^{\phantom{1}1} & \ldots & v_1^{\phantom{1} n}
\\
\vdots & \vdots & \vdots
\\
v_k^{\phantom{k}1} & \ldots & v_k^{\phantom{k}n}
\end{array}
\right)
\, ,
\end{equation}
which represents the $k$-frame. 
Switching to a different basis in $\Pi$, $v \rightarrow w$, $w_i=M_{i}^{\phantom{i}j}v_j$, with 
$M \in GL(k,\mathbb{C})$, leads to $V \rightarrow W=MV$ --- both $W$ and $V$ characterize the same $k$-plane. A standard form $\tilde{V}$ for $V$ may be chosen by taking $M$ above to be the inverse of the matrix defined by the first $k$ columns of $V$, then $\tilde{V}$ has a unit $k \times k$ matrix in that same position, and the rest of its entries, call them $m_{ij}$, $1 \leq i \leq k$, $1 \leq j \leq \tilde{k}$, where 
$\tilde{k} \equiv 2s+1-k$ is the codimension of $\Pi$, may be used as local coordinates on $\text{Gr}_{k,n}$,
\begin{equation}
\label{tildeVdef}
\tilde{V}
=
\left(
\begin{array}{ccccccc}
1 & 0 & \ldots & 0 & m_{11} & \ldots & m_{1\tilde{k}}
\\
\vdots &  &  & \vdots & \vdots & \vdots & \vdots 
\\
0 & 0 & \ldots & 1 & m_{k1} & \ldots & m_{k\tilde{k}}
\end{array}
\right)
\, ,
\end{equation}
in accordance with the (complex) dimension of $\text{Gr}_{k,n}$ being $k\tilde{k}$. 
Denote by $V^{\rha{I}}$ the minor $\Delta_{\rha{I}}$ of $V$, formed by the columns $\rha{I}=(i_1,\ldots,i_k)$ of $V$, with $1 \leq i_1< \ldots <i_k \leq  n$. Extend, for later convenience, this definition to arbitrary $k$-indices $I$ by total antisymmetry, \eg, $V^{(21)}=-V^{(12)}$, $V^{(11)}=0$, \etc. The set of all $\binom{n}{k}$ numbers $V^{\rha{I}}$ constitutes the \emph{Pl\"ucker coordinates} of the frame $V$  in $\mathbb{C}^{\binom{n}{k}}$. These are also projective coordinates for the $k$-plane $\Pi$, given that a change of basis 
$v \rightarrow w$ in $\Pi$, as above,  leads to $V^{\rha{I}}\rightarrow W^{\rha{I}}=\det(M) V^{\rha{I}}$. Thus, the plane $\Pi$ is mapped to a complex line in $\mathbb{C}^{\binom{n}{k}}$, \ie, a point in the projective space $\mathbb{P}^{\binom{n}{k}-1}$ --- this is the \emph{Pl\"ucker embedding} of $\text{Gr}(k,n)$ in $\mathbb{P}^{\binom{n}{k}-1}$. Note that a $k$-plane may be thought of as an equivalence class of $k$-frames, two frames being equivalent when their corresponding matrices are related by an invertible matrix, like $V$ and $W$  above. Accordingly, we write $\Pi=\pket{v_1, \ldots,v_k}=\pket{w_1,\ldots,w_k}=\pket{V}=\pket{W}$.

The above may be recast in a tighter language by considering the $k$-th exterior power of $\mathbb{C}^n$, $\wedge^k \mathbb{C}^n$, which, given a basis $\{e_1,\ldots,e_n\}$ of $\mathbb{C}^n$, inherits naturally the basis $\{ e_{\rha{A}} =e_{a_1}\wedge \ldots \wedge e_{a_k} \}$, with $1 \leq a_1 < \ldots < a_k \leq n$. One may then associate to the $k$-frame $V=\{v_i \}$ in $\Pi$ the $k$-vector $\mathbf{V}=v_1 \wedge \ldots \wedge v_k \in \wedge^k \mathbb{C}^n$. The Pl\"ucker coordinates defined above are just the components of this vector in the natural basis, 
\begin{equation}
\label{Plcoord}
\mathbf{V}=v_1 \wedge \ldots \wedge v_k=\sum_{\rha{I}} V^{\rha{I}} e_{\rha{I}}
\, .
\end{equation}
In terms of these vectors, a change of basis, as above, gives $\mathbf{W} \equiv w_1 \wedge \ldots \wedge w_k=\det(M) \, \mathbf{V}$, so that $\Pi$ may be identified with the ray $[\mathbf{V}]$ generated by $\mathbf{V}$ in 
$\wedge^k \mathbb{C}^n$. In the case of oriented planes, one must restrict $\det(M) >0$, and then $\Pi$ is only identified with half of the ray.

 Note that a general element  $\mathbf{P}=\sum_{\rha{I}} P^{\rha{I}} e_{\rha{I}}  \in \wedge^k \mathbb{C}^n$ is not \emph{factorizable} (or \emph{decomposable}), \ie, it cannot be written as a single $k$-fold wedge product --- the necessary and sufficient condition for factorizability is that the $P^{\rha{I}}$ satisfy the following quadratic (Pl\"ucker) relations (see, \eg, Ch.{} 1.5 of~\cite{Gri.Har:78}, Ch.{} 10.2 of~\cite{Sha.Rem:13}, Ch.{} 3.4 of~\cite{Jac:10}),
\begin{equation}
\label{Plurel}
\sum_{m=1}^{k+1} (-1)^m \,  P^{(i_1 \ldots i_{k-1} j_m)} \,  P^{(j_1 \ldots \widehat{j_{m}} \ldots j_{k+1})}=0
\, ,
\end{equation}
for all ordered multiindices $\rha{I}=(i_1,\ldots,i_{k-1})$, $\rha{J}=(j_1,\ldots,j_{k+1})$, where a hat above an index denotes omision of that index --- this is the analytical form of the Pl\"ucker embedding. Note that in writing out explicitly the above relations, one encounters, in general, coordinates $P^L$, with the multiindex $L$ not necessarily ordered, or with repeated indices --- in that case, one uses the antisymmetry mentioned above to achieve the proper ordering, or put the term equal to zero, respectively. 

Given a hermitian inner product $\ipc{\cdot}{\cdot}$ in $\mathbb{C}^n$, one may extend it to $k$-frames by
\begin{equation}
\label{ipkf}
\ipc{V}{W} = \det 
\left(
\begin{array}{ccc}
\ipc{v_1}{w_1} & \ldots & \ipc{v_1}{w_k}
\\
\vdots & \ldots & \vdots
\\
\ipc{v_k}{w_1} & \ldots & \ipc{v_k}{w_k}
\end{array}
\right)
\, ,
\end{equation}
which gives rise to the following  inner product between two $k$-planes $\Pi=\pket{V}$, $\Sigma=\pket{W}$,
\begin{equation}
\label{ipkp}
\ipc{\Pi}{\Sigma}=
\frac{|\ipc{V}{W}|}{\sqrt{\ipc{V}{V}}\sqrt{\ipc{W}{W}}}
\, .
\end{equation}

\section{The principal constellation of a spin-$s$ $k$-plane}
\label{kpCvHf}
Our first attempt at a stellar representation of a spin-$s$ $k$-plane (henceforth an \emph{$(s,k)$-plane}),  generalizes the view of the Majorana polynomial $P_{\ket{\Psi}}(\zeta)$ of a state $\ket{\Psi}$ as the polynomial part of the inner product $\braket{-n}{\Psi}$, where $n=(\theta,\phi)$ and $\zeta=\tan\frac{\theta}{2}e^{i\phi}$, which results in the stars of $\ket{\psi}$ being antipodal to the zeros of its Husimi function $H_{\ket{\psi}}(n)=|\braket{n}{\psi}|^2$. To this end, we need to generalize the concept of a spin-$s$ coherent state to that of a \emph{coherent $(s,k)$-plane}. 
\begin{definition}
\label{SEVdef}
For a general $(s,k)$-plane $\Pi=[\ket{\psi_1},\ldots,\ket{\psi_k}]$, with $\braket{\psi_\mu}{\psi_\nu}=\delta_{\mu \nu}$, $\mu,\nu=1,\ldots,k$, we define its \emph{spin expectation value} (SEV) $\langle \mathbf{S} \rangle_\Pi$ to be a vector in physical $\mathbb{R}^3$, with components $\langle S_i \rangle_\Pi$, $i=1,2,3,$ given by
\begin{equation}
\label{SEVPi}
\langle S_i \rangle_{\Pi}=\Tr 
\left(
\begin{array}{ccc}
\bra{\psi_1} S_i \ket{\psi_1} 
&
\ldots
&
\bra{\psi_1} S_i \ket{\psi_k}
\\
\vdots
&
\vdots
&
\vdots
\\
\bra{\psi_k} S_i \ket{\psi_1} 
&
\ldots
&
\bra{\psi_k} S_i \ket{\psi_k}
\end{array}
\right)
\, .
\end{equation}
\end{definition}
\begin{definition}
\label{coherentkplane_def}
An $(s,k)$-plane $\Pi$ is coherent if the modulus of its SEV is maximal among all $(s,k)$-planes.
\end{definition}
As the following theorem shows, the space of coherent $(s,k)$-planes is not different from that of the spin coherent states. 
\begin{theorem}
\label{ckp_thm}
Coherent $(s,k)$-planes are in 1-1 correspondence with unit vectors in $S^2 \subset \mathbb{R}^3$. For a given such vector $n$, the coherent $(s,k)$-plane along $n$, denoted by $\Pi_n$, is given by $\Pi_n=[\ket{n,s},\ket{n,s-1},\ldots,\ket{n,s-k+1}]$ with maximal SEV modulus $|\langle \mathbf{S} \rangle_{\Pi_n}|=\frac{k}{2}(2s+1-k)$.
\end{theorem}
\begin{proof}
It is easily shown that rotating the kets $\ket{\psi_\mu}$, that define $\Pi$, by $D^{(s)}(R)$ results in a rotation of $\langle \mathbf{S} \rangle_\Pi$ by $R \in SO(3)$. We may then assume, without loss of generality, that  $\langle \mathbf{S} \rangle_\Pi$ is  along $\hat{z}$, so that $|\langle \mathbf{S} \rangle_\Pi|=|\langle S_z \rangle_\Pi|=|\bra{\psi_1}S_z\ket{\psi_1}+\ldots+\bra{\psi_k}S_z\ket{\psi_k}|$. $S_z$ acts on wedge products as a derivation, so, for $\Pi=[\ket{\psi_1},\ldots,\ket{\psi_k}]$, with the $\ket{\psi_\mu}$ orthonormal, 
\begin{equation}
\label{TrSzwedge}
\langle S_z\rangle_\Pi
=
\bra{\psi_1}\wedge \ldots \wedge \bra{\psi_k} S_z \ket{\psi_1}\wedge \ldots \wedge \ket{\psi_k}
\, ,
\end{equation}
where the inner product of $k$-fold wedge products is $k!$ times that of the corresponding $k$-frames (see~(\ref{ipkf})). It is clear that the \rhs{} of~(\ref{TrSzwedge}) is maximized when $\ket{\psi_1}\wedge \ldots \wedge \ket{\psi_k}$ is the eigenvector of $S_z$ with the maximal eigenvalue, \ie, $\ket{s,s}\wedge \ldots \wedge \ket{s,s-k+1}$, with eigenvalue $s+(s-1)+\ldots +(s-k+1)=\frac{1}{2}k \tilde{k}$.
\end{proof}
\begin{theorem}
\label{statesinPin}
Every state $\ket{\psi} \in \Pi_n$ has at least $\tilde{k}$ stars along $n$.
\end{theorem}
\begin{proof}
Every $\ket{\psi}$ in $\Pi_n$ is a linear combination of the states $\Pi_n$ factorizes into, therefore, the Majorana polynomial of $\ket{\psi}$ is the same linear combination of the Majorana polynomials of those states. But the latter all have  at least $\tilde{k}$ stars along $n$, property that is easily seen to be inherited by $P_{\ket{\psi}}$.
\end{proof}
We may now define the \emph{principal constellation} of an $(s,k)$-plane $\Pi=[W]$ as the set of those stars $n$ 
(counted with multiplicity) for which  
$\ipc{\Pi_{-n}}{\Pi}=0$ --- the corresponding polynomial, \ie, the one whose roots are the stereographic projections of those stars, will be the \emph{ principal polynomial} $P_\Pi(\zeta)$ of $\Pi$;
formally, it is defined by
\begin{equation}
\label{PPiVW}
P_\Pi(\zeta)=\zeta^{k \tilde{k}} \ipc{\tilde{V}_{-n}}{W}
\, ,
\end{equation}
where $\Pi_n=[\tilde{V}_{n}]$, $\tilde{V}_{-n}$ is the standard representative of its class (see~(\ref{tildeVdef})), 
 and $\zeta$ is related to $n$ in the standard way.
\begin{theorem}
\label{starspsiPi}
A star $n \in S^2$ belongs to the constellation of  an $(s,k)$-plane $\Pi$ if and only if  there exists a state $\ket{\psi}\in \Pi$ the constellation of which has at least $k$ stars along $n$.
\end{theorem}
\begin{proof}
A star $n$ belongs to the constellation $C_\Pi$ of an$(s,k)$-plane $\Pi$ iff $\ipc{\Pi_{-n}}{\Pi}=0$. When two $k$-planes are orthogonal, there exists in each of them a vector that is orthogonal to all the vectors of the other. Thus, there is a state $\ket{\psi} \in \Pi$ that is orthogonal to all the states in $\Pi_{-n}$, and belongs, therefore, to the orthogonal complement $\Pi_{-n}^\perp$ of $\Pi_{-n}$. The latter is easily seen to be a coherent $(s,\tilde{k})$-plane along $n$, so that, due to theorem~\ref{statesinPin}, $\ket{\psi}$ has at least $\tilde{\tilde{k}}=k$ stars along $n$.
\end{proof}
We state at this point that the degree of $P_\Pi(\zeta)$, for an $(s,k)$-plane $\Pi$, is $k\tilde{k}$. There are various ways to see this --- a simple one is given in Corollary~\ref{degreeMajoPoly_thm} below.
Thus, $(s,k)$-planes have Majorana constellations of $k \tilde{k}$ stars, some of which may coincide. Just like in the original Majorana polynomial, if $P_\Pi(\zeta)$ turns out to be of a lower degree, the missing roots are taken to be at infinity, so that the missing stars of the constellation are put at the south pole.
\begin{myexample}{A tetrahedral $(\frac{3}{2},2)$-plane}{ex:s32k2}
Denote by $\{e_i, \, i=1,\ldots,4\}$ the orthonormal $S_z$-eigenbasis in the spin-3/2 Hilbert space $\mathbb{C}^4$,
\begin{equation}
\label{s32eb}
\{e_1,e_2,e_3,e_4\}
=
\left\{
\ket{\frac{3}{2},\frac{3}{2}},
\ket{\frac{3}{2},\frac{1}{2}},
\ket{\frac{3}{2},-\frac{1}{2}},
\ket{\frac{3}{2},-\frac{3}{2}}
\right\}
\, .
\end{equation}
The induced basis in $\wedge^2 \mathbb{C}^4$ is
$\{ e_{12},e_{13},e_{14},e_{23},e_{24},e_{34} \}$, with $e_{ij} \equiv e_i \wedge e_j$. The coherent $(\frac{3}{2},2)$-plane along $z$ is $\Pi_z=e_{12}$, with corresponding matrix 
\begin{equation}
\label{Vcohz}
V_z=
\left(
\begin{array}{cccc}
1 & 0 & 0 & 0
\\
0 & 1 & 0 & 0
\end{array}
\right)
\, ,
\end{equation}
which is already in the standard form, so that, in this chart of  $\text{Gr}_{2,4}$, $\Pi_z$ is at the origin (the four rightmost entries of $V_z$ are zero).
We may rotate $\Pi_z$ to a general direction $n=(\theta,\phi)$ to obtain $\Pi_n$ (using, \eg, the ``geodesic'' rotation $R_{(-\sin\phi,\cos\phi,0),\theta}$) --- the corresponding matrix is
\begin{equation}
\label{Pins32k2}
V_n
=
\left(
\begin{array}{cccc}
 \cos ^3\left(\frac{\theta }{2}\right)
   & -\frac{1}{4} \sqrt{3} e^{i \phi }
   \csc \left(\frac{\theta }{2}\right)
   \sin ^2(\theta ) & \frac{1}{2}
   \sqrt{3} e^{2 i \phi } \sin
   \left(\frac{\theta }{2}\right) \sin
   (\theta ) & -e^{3 i \phi } \sin
   ^3\left(\frac{\theta }{2}\right) \\
 \frac{1}{4} \sqrt{3} e^{-i \phi }
   \csc \left(\frac{\theta }{2}\right)
   \sin ^2(\theta ) & \frac{1}{4}
   \left(\cos \left(\frac{\theta
   }{2}\right)+3 \cos \left(\frac{3
   \theta }{2}\right)\right) &
   \frac{1}{4} e^{i \phi } \left(\sin
   \left(\frac{\theta }{2}\right)-3
   \sin \left(\frac{3 \theta
   }{2}\right)\right) & \frac{1}{2}
   \sqrt{3} e^{2 i \phi } \sin
   \left(\frac{\theta }{2}\right) \sin
   (\theta ) \\
\end{array}
\right)
\, ,
\end{equation}
which, brought to the standard form, becomes
\begin{equation}
\label{Pins32k2sta}
\tilde{V}_n
=
\left(
\begin{array}{cccc}
 1 & 0 & -\sqrt{3} \zeta^2 & 2 \zeta^3
 \\
 0 & 1 & -2\zeta & \sqrt{3} \zeta^2
\end{array}
\right)
\, ,
\end{equation}
where $\zeta=\tan \frac{\theta}{2} e^{i \phi}$ is the stereographic image of $n$.

Consider now, as an example, the  $(\frac{3}{2},2)$-plane $\Pi_{\text{tetra}}$, with standard representative
\begin{equation}
\label{Wtetra}
\tilde{W}_{\text{tetra}}
=
\left(
\begin{array}{cccc}
1 & 0 & 0 & \sqrt{2}
\\
0 & 1 & 0 & 0
\end{array}
\right)
\, ,
\end{equation}
and compute, using~(\ref{ipkf}), (\ref{PPiVW}),
\begin{equation}
\label{ipPinW}
P_{\Pi_{\text{tetra}}}(\zeta)=\zeta^4 \ip{V_{-n}}{\tilde{W}_{\text{tetra}}}
=
\zeta^4-2\sqrt{2}\zeta
\, ,
\end{equation}
where we used the fact that the stereographic image of $-n$ is $-1/\bar{\zeta}$, $\bar{\zeta}$ denoting the complex conjugate of $\zeta$. Note that this coincides with the Majorana polynomial of the tetrahedral state, in Example~\ref{ex:s2k1} --- we conclude that the principal constellation of $\Pi_{\text{tetra}}$ is the same regular tetrahedron found there.
\end{myexample}

The above definition of $P_\Pi(\zeta)$, while quite analogous to that of the standard Majorana polynomial, turns out to be rather awkward to work with, as it typically involves the computation of large rotation matrices, which take $\Pi_z$ to $\Pi_n$. It also fails to shed any light to the natural question of the relation between the principal polynomial of a plane, $P_\Pi(\zeta)$, and those of the states it factorizes into, $\{ P_{\ket{\psi_\mu}}(\zeta), \, \mu=1,\ldots,k\}$.
Both shortcomings are bypassed by the following
\begin{theorem}
\label{MajoPolyW}
The principal polynomial $P_{\Pi}(\zeta)$ of an $(s,k)$-plane $\Pi=\ket{\psi_1}\wedge \ldots \wedge \ket{\psi_k}$ is given by the Wroskian of the Majorana polynomials $P_{\ket{\psi_\mu}}(\zeta)$ of the states $\ket{\psi_\mu}, \mu=1,\ldots,k$, \ie,
\begin{equation}
\label{WronskiMap}
P_\Pi(\zeta)=
\det
\left(
\begin{array}{cccc}
P_{\ket{\psi_1}}(\zeta) 
&
P'_{\ket{\psi_1}}(\zeta) 
&
\ldots
&
P^{(k-1)}_{\ket{\psi_1}}(\zeta) 
\\
\vdots
&
\vdots
&
\vdots
&
\vdots
\\
P_{\ket{\psi_k}}(\zeta) 
&
P'_{\ket{\psi_k}}(\zeta) 
&
\ldots
&
P^{(k-1)}_{\ket{\psi_k}}(\zeta) 
\end{array}
\right)
\, ,
\end{equation}
where $P'(\zeta)\equiv \partial P/\partial\zeta$ and $P^{(r)}(\zeta)\equiv \partial^r P/\partial\zeta^r$.
\end{theorem}
\begin{proof}
Consider a star $n$ in the constellation of $\Pi$ and call $\zeta_0$ its stereographic image. By theorem~\ref{starspsiPi} this only happens iff there exists a state $\ket{\psi} \in \Pi$ the constellation of which has at least $k$ stars along $n$, so that 
$P_{\ket{\psi}}$ has $\zeta_0$ as a $k$-fold root. But $\ket{\psi} \in \Pi$ implies that $P_{\ket{\psi}}$ can be written as a linear combination of $P_{\ket{\psi_\mu}}$, $\mu=1, \ldots,k$,
\begin{equation}
\label{Ppsimu}
P_{\ket{\psi}}=\sum_{\mu=1}^k c_\mu P_{\ket{\psi_\mu}}
\, .
\end{equation}
$\zeta_0$ being a $k$-fold root of $P_{\ket{\psi}}$ is equivalent to it being a root of $P_{\ket{\psi}}$  and of all its first $k-1$ derivatives,
\begin{equation}
\label{zetamultiroot}
\sum_{\mu=1}^k c_\mu P_{\ket{\psi_\mu}}(\zeta_0)
=
0
\, ,
\qquad
\sum_{\mu=1}^k c_\mu \frac{\partial P_{\ket{\psi_\mu}}}{\partial \zeta}(\zeta_0)
=
0
\, ,
\quad
\ldots
\, ,
\qquad
\sum_{\mu=1}^k c_\mu  P^{(k-1)}_{\ket{\psi_\mu}}(\zeta_0)
=
0
\, ,
\end{equation}
where $P^{(r)}_{\ket{\psi}}(\zeta) \equiv \partial^r P_{\ket{\psi}}(\zeta)/\partial\zeta^r$. The above equations define a linear system in the unknowns $c_\mu$, which has a nontrivial solution iff the determinant of its coefficients is zero.
\end{proof}
The map from the Majorana polynomials of the states to the principal polynomial of the plane given in~(\ref{WronskiMap}) is known as a \emph{Wronski map} and plays an important role in  algebraic geometry, combinatorics, and control theory (see, \eg, \cite{Ere.Gab:02}).

Since an $(s,k)$-plane $\Sigma$ and its orthogonal complement $\Sigma^\perp$ carry the same geometrical information, one expects their constellations to be related.
\begin{theorem}
\label{Pianti}
The principal constellations of an $(s,k)$-plane and its orthogonal complement are antipodal to each other. 
\end{theorem}
\begin{proof}
We denote here explicitly the dimension of the planes by superindices in parentheses. A star $n$ is in the constellation of $\Sigma^{(k)}$ iff $\ipc{\Pi^{(k)}_{-n}}{\Sigma^{(k)}}=0$. If two $k$-planes are orthogonal, their orthogonal complements also are, and $(\Pi^{(k)}_{-n})^\perp=\Pi^{(\tilde{k})}_n$, so that $\ipc{\Pi^{(\tilde{k})}_{n}}{(\Sigma^{(k)})^\perp}=0$, and the assertion follows.
\end{proof}

So far we have specified how to assign to an $(s,k)$-plane a unique constellation of $k\tilde{k}$ stars. The natural question that arises is whether this map is 1-to-1. 
Note that the number of stars in the constellation  coincides with the complex dimension of 
$\text{Gr}_{k,n}$, which sounds encouraging.  However, some experimentation quickly leads to the conclusion that this is not the case.
\begin{myexample}{Two $(\frac{3}{2},2)$-planes with the same principal constellation}{Ts322planes}
Define a generic $(\frac{3}{2},2)$-plane $\Sigma=[\tilde{W}]$ by its standard representative,
\begin{equation}
\label{Wgen}
\tilde{W}=
\left(
\begin{array}{cccc}
1 & 0 & m_{11} & m_{12}
\\
0 & 1 & m_{21} & m_{22}
\end{array}
\right)
\, ,
\end{equation}
and compute its principal polynomial (using~(\ref{PPiVW}), (\ref{Pins32k2sta})),
\begin{equation}
\label{majoPolyWgen}
P_\Sigma(\zeta)=\zeta^4 +2 m_{21} \zeta^3+\sqrt{3}(m_{22}-m_{11}) \zeta^2 -2 m_{12} \zeta
+m_{11}m_{22}-m_{12}m_{21}
\, .
\end{equation}
Take now a particular fourth degree polynomial, say, $\zeta^4-1$, the roots of which define a square on the equator of the Riemann sphere, and set it equal to $P_\Sigma$ to find two solutions
\begin{equation}
\label{twosols32k2}
\tilde{W}_1=
\left(
\begin{array}{cccc}
1 & 0 & i & 0
\\
0 & 1 & 0 & i
\end{array}
\right)
\, ,
\qquad 
\tilde{W}_2=
\left(
\begin{array}{cccc}
1 & 0 & -i & 0
\\
0 & 1 & 0 & -i
\end{array}
\right)
\, ,
\end{equation}
which are actually orthogonal to each other, $\ipc{\tilde{W}_1}{\tilde{W}_2}=0$. 
This is as expected from theorem~\ref{Pianti}, as the constellation considered is self-antipodal. 
\end{myexample}

Further similar computations reveal that, generically, there are 2 $(\frac{3}{2},2)$-planes that share the same  4-star constellation, while, for example,  there are 5 $(2,3)$-planes sharing the same 6-star constellation. Initial attempts to discern a pattern in these numbers were quickly shown hopeless: as we are about to prove, there are, generically, exactly 1,662,804 $(4,4)$-planes sharing the same 20-star constellation, and, for larger $s$, the numbers simply explode. A sense of order is restored by the following
\begin{theorem}
\label{numberofkplanes}
The number $Q(s,k)$ of $(s,k)$-planes that, generically, share the same principal constellation, is given by
\begin{equation}
\label{Nsk}
Q(s,k)
=
\frac{1! \, 2! \,3! \,\dots (k-1)!}{\tilde{k}! \, (\tilde{k}+1)! \dots (2s)!} ( k\tilde{k})! 
\, .
\end{equation}
\end{theorem}
\begin{proof}
As shown in the proof of theorem~\ref{starspsiPi}, if a star $n$ is in the constellation of $\Pi^{(k)}$, then there exists a state $\ket{\psi} \in \Pi^{(k)}$ that also belongs to $(\Pi^{(k)}_{-n})^\perp=\Pi^{(\tilde{k})}_n$, \ie, $\Pi$ intersects (nontrivially) $\Pi^{(\tilde{k})}_n$. Then, given the $k\tilde{k}$ stars $n_i$ of the constellation of $\Pi$, the number of $k$-planes that share that same constellation is the number of $k$-planes that intersect (nontrivially) the $k\tilde{k}$ $\tilde{k}$-planes $\Pi^{(\tilde{k})}_{n_i}$. This number has been shown by Schubert~\cite{Sch:79a} to be equal, generically,  to $Q(s,k)$ above (see also \cite{Gri.Har:78} or Ch.{} XIV of \cite{Hod.Ped:52} for a modern treatment).
\end{proof}
Note that the result applies to the generic case --- particular constellations might have fewer corresponding $k$-planes, for instance,  the tetrahedral constellation in Example~\ref{Ts322planes} has only one corresponding 2-plane, rather than two (= $Q(3/2,2)$). What transpires in these cases is that as one approaches the constellation in question, two or more corresponding planes approach each other, and become identical right on the constellation. Thus, if the planes are counted with multiplicities, their number is always $Q(s,k)$ above. In any case, theorem~\ref{numberofkplanes} makes it clear that the principal constellation of an $(s,k)$-plane, as defined above, does not uniquely characterize that plane ---  it turns out that what is missing is more constellations.
\section{$(s,k)$-plane Multiconstellations via the Pl\"ucker Embedding}
\label{kpMvtPE}
\subsection{The $SU(2)$ action on $\wedge^k \mathcal{H}$}
\label{TSao}
A spin-$s$ quantum state $\ket{\psi}$ lives in the Hilbert space $\mathcal{H}=\mathbb{C}^{\tilde{N}}$ --- its image in the projective space $\mathbb{P}^N$ will be denoted by $\pket{\psi}$. We transcribe the general notation we have used so far to the case at hand: vectors are denoted by kets, like $\ket{\psi}$, and $k$-planes in $\mathcal{H}$ can be described as collections of vectors, $\{ \ket{\psi_1}, \ldots,\ket{\psi_k} \}$, the $k \times n$ matrix $\Psi$ of their components, or their wedge product $\ket{\boldsymbol{\Psi}}$, itself a vector in $\wedge^k \mathcal{H}$, $\ket{\boldsymbol{\Psi}}=\sum_{\rha{I}}\Psi^{\rha{I}} e_{\rha{I}}$.

Rotations in physical space are represented by the action of $SU(2)$ on $\mathcal{H}$ via $g \mapsto D^{(s)}(g)$, where $D^{(s)}$ is the $\tilde{N}$-dimensional irreducible representation of $SU(2)$, \ie, under a rotation $g$, the column vector $\ket{\psi}$ transforms by left multiplication by $D^{(s)}(g)$, $\ket{\psi} \mapsto D^{(s)}(g) \ket{\psi}$.  This representation extends naturally to $\wedge^k \mathcal{H}$ by tensoring up,
\begin{equation}
\label{repwedgek}
\ket{\boldsymbol{\Psi}}=
\ket{\psi_1} \wedge \ldots \wedge \ket{\psi_k} \mapsto D^{(s)}(g) \ket{\psi_1} \wedge \ldots \wedge D^{(s)}(g) \ket{\psi_k}
\equiv
D^{(s,k)}(g) \ket{\boldsymbol{\Psi}}
\, ,
\end{equation}
where $D^{(s,k)}(g)$, \ie, the totally antisymmetric part of the $k$-th tensor power of $D^{(s)}(g)$, provides a $\binom{\tilde{N}}{k}$-dimensional representation of $SU(2)$ on $\wedge^k\mathcal{H}$ and, with a slight abuse of notation, $\ket{\boldsymbol{\Psi}}$ on the right hand side stands for the column vector of the components of $\ket{\psi_1}\wedge \ldots \wedge \ket{\psi_n}$ in the \emph{Pl\"ucker basis} of the $e_{\rha{I}}$'s. This representation is not, in general, irreducible: when brought in block-diagonal form, by a suitable change of basis in $\wedge^k \mathcal{H}$, from the Pl\"ucker to the Block Diagonal (BD) one, $D^{(s,k)}$ contains $m^{(s,k)}_j$ copies of $D^{(j)}$, $j=0, \ldots, s_{\text{max}}$. We turn now to the determination of the BD basis, as well as of $s_{\text{max}}$ and the multiplicities $m^{(s,k)}_j$.

As mentioned above, $SU(2)$ acts on wedge products by its tensor powers, giving rise to the representation $D^{(s,k)}$. At the Lie algebra level, this implies that the generators $S_a$, $a=1,2,3$, act as derivations, \ie, by following Leibniz' rule, which results in representation matrices $S^{(s,k)}_a$, satisfying the $\mathfrak{su}(2)$ algebra, and generating $D^{(s,k)}$ by exponentiation,
\begin{equation}
\label{Sskdef}
S^{(s,k)}_a=i \frac{\partial}{\partial t} D^{(s,k)}(e^{-i t S_a})|_{t=0}
\, .
\end{equation}
As a result, a wedge product of, say, $S^{(s)}_z$ eigenvectors, is a $S^{(s,k)}_z$ eigenvector, with eigenvalue equal to the sum of the eigenvalues of the factors. Thus, the ``top'' $(s,k)$-plane 
$\ket{s,s} \wedge \ket{s,s-1} \wedge \ldots \wedge \ket{s,s-(k-1)}$ attains the maximal $S^{(s,k)}_z$ eigenvalue, which is also the maximal value of the spin $j$ in the decomposition of $D^{(s,k)}$, equal to 
\begin{equation}
\label{smaxdef}
s_\text{max}=s+ (s-1)+ \ldots +(s-(k-1))= \frac{1}{2} k \tilde{k}
\, .
\end{equation} 
We denote the above $(s,k)$-plane by $\ket{s_\text{max},s_\text{max}}$. Looking for $(s,k)$-planes with $S^{(s,k)}_z$-eigenvalue equal to $s_\text{max}-1$, we find only one, $\ket{s_\text{max},s_\text{max}-1}=\ket{s,s}\wedge \ldots \wedge \ket{s,s-(k-2)} \wedge \ket{s, s-k}$, which is obtained from $\ket{s_\text{max},s_\text{max}}$ by applying the lowering operator $S^{(s,k)}_-$, \ie, it belongs to the same irreducible representation. Going one step down, one finds two new eigenvectors with eigenvalue $s_\text{max}-2$. A linear combination of them is obtained as $S^{(s,k)}_- \ket{s_\text{max},s_\text{max}-1}$, while the othogonal combination serves as the heighest weight vector of a $j=s_\text{max}-2$ irreducible multiplet. We conclude that for all $s$, $k$, the representations with $j=s_\text{max}$ and $j=s_\text{max}-2$ appear with multiplicity 1, while $j=s_\text{max}-1/2$, $s_\text{max}-1$, $s_\text{max}-3/2$ never appear. Continuing in the same way, one may construct the $\wedge^k \mathcal{H}$-basis that block diagonalizes $D^{(s,k)}$. If, however, only the multiplicities $m^{(s,k)}_j$ are desired, one may employ the standard character machinery~\cite{Lit:50,Ful.Har:04}.
Thus, one first computes  the character 
\begin{align}
\label{charsk}
\chi^{(s,k)}(\alpha)
&
\equiv
\Tr \, D^{(s,k)}(R_{\hat{n},\alpha})
=
 \sum_{j=0}^{s_{\text{max}}} m^{(s,k)}_j \chi^{(j)}(\alpha)
 \, ,
 \end{align}
where $\hat{n}$ denotes the rotation axis, and $\alpha$  the rotation angle, and $\chi^{(j)}$ are the irreducible characters, 
\begin{equation}
\label{irredjchar}
\chi^{(j)}(\alpha) \equiv \Tr \, D^{(j)}(R_{\hat{n},\alpha}) = \frac{\sin((j+\frac{1}{2})\alpha)}{\sin \frac{\alpha}{2}}
\, .
\end{equation}
Then the orthonormality of the irreducible characters is invoked, $\ipp{\chi^{(m)}}{\chi^{(n)}}=\delta_{mn}$, where
\begin{equation}
\label{ippdef}
\ipp{f}{h}
\equiv
\frac{1}{\pi}\int_0^{2\pi} \td \alpha \, \sin^2 \frac{\alpha}{2} \bar{f}(\alpha) h(\alpha)
\, ,
\end{equation}
to get for the multiplicities
\begin{equation}
\label{multires}
m^{(s,k)}_j=\frac{1}{\pi} \int_0^{2\pi} \td \alpha \,  \sin^2 \frac{\alpha}{2} \chi^{(s,k)}(\alpha) \chi^{(j)}(\alpha)
\, .
\end{equation}
The characters $\chi^{(s,k)}$ satisfy the recursion formula
\begin{equation}
\label{chiskrec}
\chi^{(s,k)}(\alpha)=
\frac{1}{k} \sum_{m=1}^k (-1)^{m-1} \chi^{(s)}(m\alpha) \chi^{(s,k-m)}(\alpha)
\, ,
\end{equation}
with $\chi^{(s,0)}(\alpha)\equiv 1$, giving, for example,
\begin{align}
\chi^{(s,2)}(\alpha)
&=
\frac{1}{2} 
\left(
\chi^{(s)}(\alpha)^2-\chi^{(s)}(2\alpha)
\right)
\label{chis2}
\\
\chi^{(s,3)}(\alpha)
&=
\frac{1}{6} 
\left(
\chi^{(s)}(\alpha)^3-3\chi^{(s)}(\alpha) \chi^{(s)}(2\alpha)+2\chi^{(s)}(3\alpha)
\right)
\label{chis3}
\\
\chi^{(s,4)}(\alpha)
&=
\frac{1}{24}
\left(
\chi^{(s)}(\alpha)^4
-6\chi^{(s)}(\alpha)^2\chi^{(s)}(2\alpha)
+3\chi^{(s)}(2\alpha)^2
+8\chi^{(s)}(\alpha)\chi^{(s)}(3\alpha)
-6\chi^{(s)}(4\alpha)
\right)
\, .
\label{chis4}
\end{align}
A general solution for the recursion~(\ref{chiskrec}) can be found, using standard representation theory machinery. Call $\lambda_m \equiv e^{i m\alpha}$, $-s\leq m \leq s$, the eigenvalues of $D^{(s)}(R_{n,\alpha})$. Then the eigenvalues of $D^{(s,k)}(R_{n,\alpha})$ are the products $\lambda_{m_1} \ldots \lambda_{m_k}$, with $m_1 < \ldots < m_k$, so that $\chi^{(s,k)}(\alpha)=\sum_{m_1 < \ldots < m_k} \lambda_{m_1} \ldots \lambda_{m_k} \equiv E_k (\boldsymbol{\lambda})$, where $E_k(\boldsymbol{\lambda})$ is the $k$-th elementary symmetric polynomial in the $2s+1$ variables $\boldsymbol{\lambda}=\{\lambda_m \}$. The latter can be expressed in terms of the Newton (or power sum) polynomials $P_r(\boldsymbol{\lambda}) = \sum_{m=-s}^s \lambda_m^r=\chi^{(s)}(r\alpha)$. To this end, given a $k$-tuple of non-negative integers $M=(m_1,\ldots,m_k)$, satisfying $\sum_{r=1}^k r m_r=k$, define the homogeneous, degree-$k$ polynomial  $P^{(M)}\equiv P_1^{m_1} \ldots P_k^{m_k}$, in terms of which (see, \eg, appendix A of~\cite{Ful.Har:04})
\begin{equation}
\label{recchi}
E_k
=
\sum_M \frac{(-1)^{k-\bar{M}}}{z(M)} P^{(M)}
\, ,
\end{equation}
where
\begin{equation}
\label{barMzM}
\bar{M}=\sum_{r=1}^k m_r
\, ,
\qquad
z(M)= m_1  ! \, 1^{m_1} m_2  ! \, 2^{m_2} \ldots m_k  ! \, k^{m_k}
\, ,
\end{equation}
so that
\begin{equation}
\label{recchi2}
\chi^{(s,k)}(\alpha)
=
\sum_M \frac{(-1)^{k-\bar{M}}}{z(M)} \left( \chi^{(s)}(\alpha) \right)^{m_1}   \left( \chi^{(s)}(2\alpha) \right)^{m_2}
\ldots  \left( \chi^{(s)}(k\alpha) \right)^{m_k}
\, .
\end{equation}
 For example, when $k=4$, the possible values of $M$ in~(\ref{recchi}) are $(4,0,0,0)$, $(2,1,0,0)$, $(0,2,0,0)$, $(1,0,1,0)$, and $(0,0,0,1)$, each of which gives rise to one of the five terms in the \rhs{} of~(\ref{chis4}).

Using these expressions, and~(\ref{irredjchar}), in~(\ref{multires}) one may compute any desired multiplicity $m^{(s,k)}_j$. Note that only integer (half-integer) values of $j$ need be considered in~(\ref{multires})  when $s_\text{max}$ is integer (half-integer). It is also clear that $s_\text{max}$ is half-integer only when $s$ is, and $k$ is odd.

The above method for obtaining the multiplicities is the standard one, but quickly becomes inefficient due to the integration in~(\ref{multires}). A much more efficient way to produce the $m^{(s,k)}_j$, based on a  combinatorial formula, is given by the following theorem~\cite{Pol.Sfe:16}.
\begin{theorem}
\label{mjskcomb}
The multiplicity $m_j^{(s,k)}$ in the \rhs{} of~(\ref{charsk}) is given by the coefficient of $x^j$, $0 \leq j \leq s_{\text{max}}$, in the Laurent expansion, around $x=0$, of the function
\begin{equation}
\label{genfunmskj}
\zeta_{s,k}(x)=(1-x^{-1}) \prod_{r=1}^k \frac{x^{s+1}-x^{r-s-1}}{x^r-1}
\, .
\end{equation}
\end{theorem}
\begin{proof}
A nice proof, based on partition function methods, can be consulted in~\cite{Pol.Sfe:16}. 
\end{proof}
We note in passing that the simple pattern that can be discerned  from Table~\ref{multi234} for $k=2$, \ie, that the series starts at $s_{\text{max}}$ and descends in steps of two, all multiplicities being one, can be shown to hold indeed true for all $s$  --- see, \eg, exercise 6.16 in~\cite{Ful.Har:04}. Already for $k=3$, the relatively tame sample in Table~\ref{multi234} does little justice to the subtle \emph{follie} unravelled, \eg, in Fig.~\ref{fig:m1003j}, where $s=40$ (left) and $s=100$ (right).
\begin{figure}[h]
\includegraphics[width=.48\linewidth]{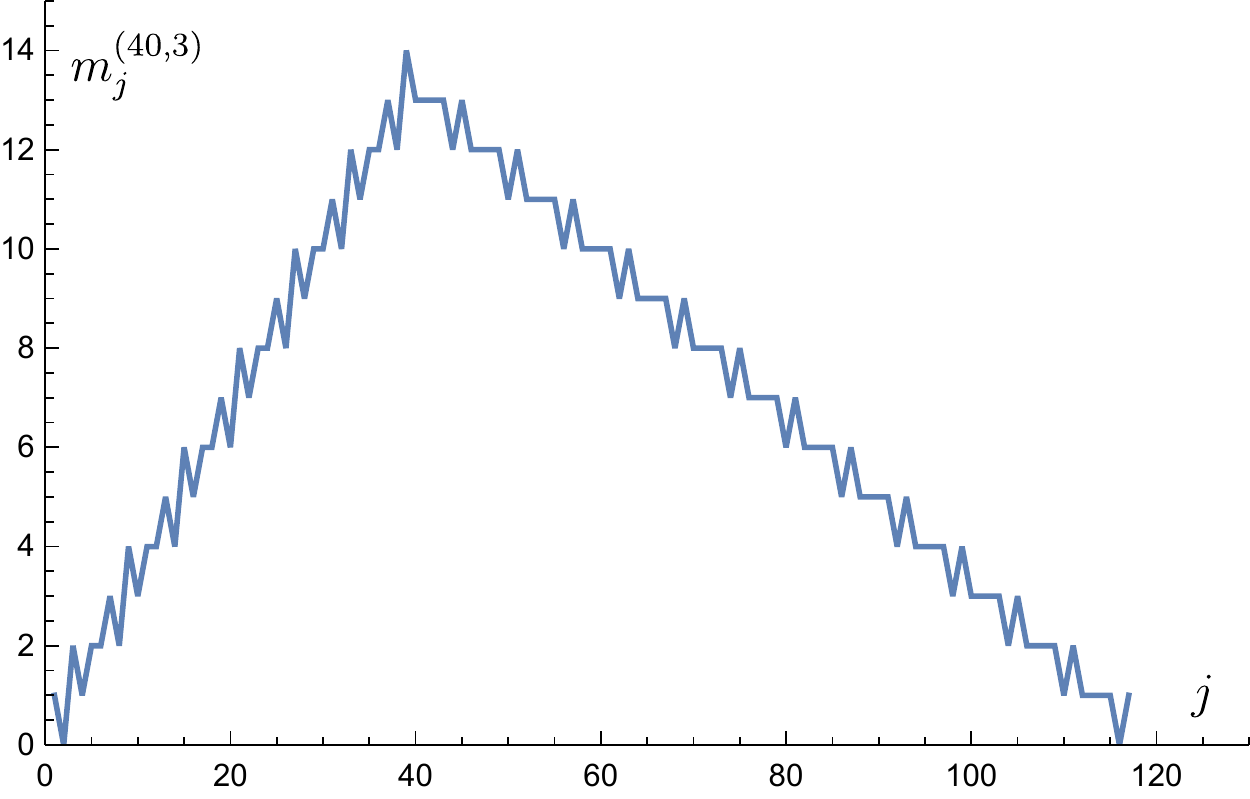}
\includegraphics[width=.48\linewidth]{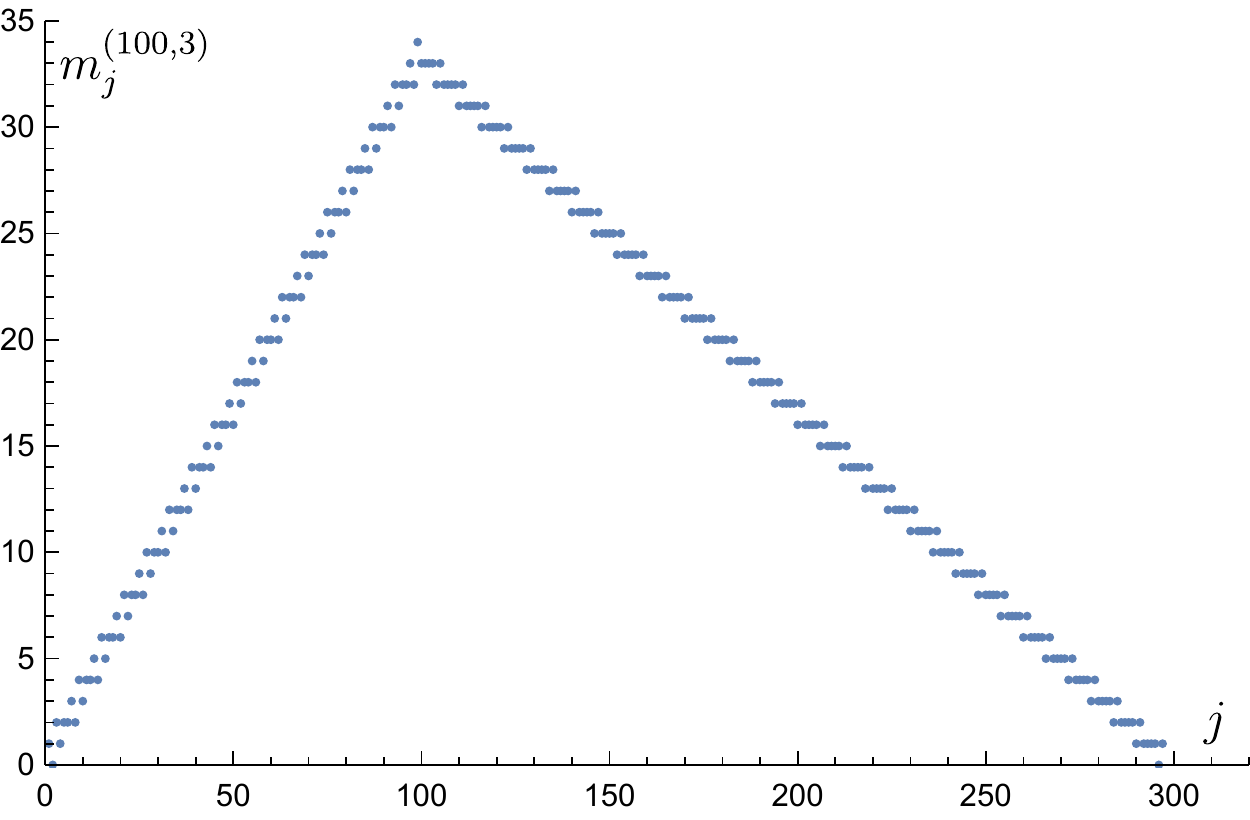}
\caption{%
Multiplicities  $m^{(40,3)}_j$ (left), $m^{(100,3)}_j$ (right), \emph{vs.} $j$,  as given by the Laurent expansion of the \rhs{} of~(\ref{genfunmskj}). In the figure on the left, consecutive points have been joined by straight  line segments to draw attention to the (locally) non-monotonic behavior of the $m$'s, which might, otherwise, pass unnoticed (a local ``period'' of 4 is easily discerned).
}
\label{fig:m1003j}
\end{figure}
\begin{figure}[h]
\includegraphics[width=.48\linewidth]{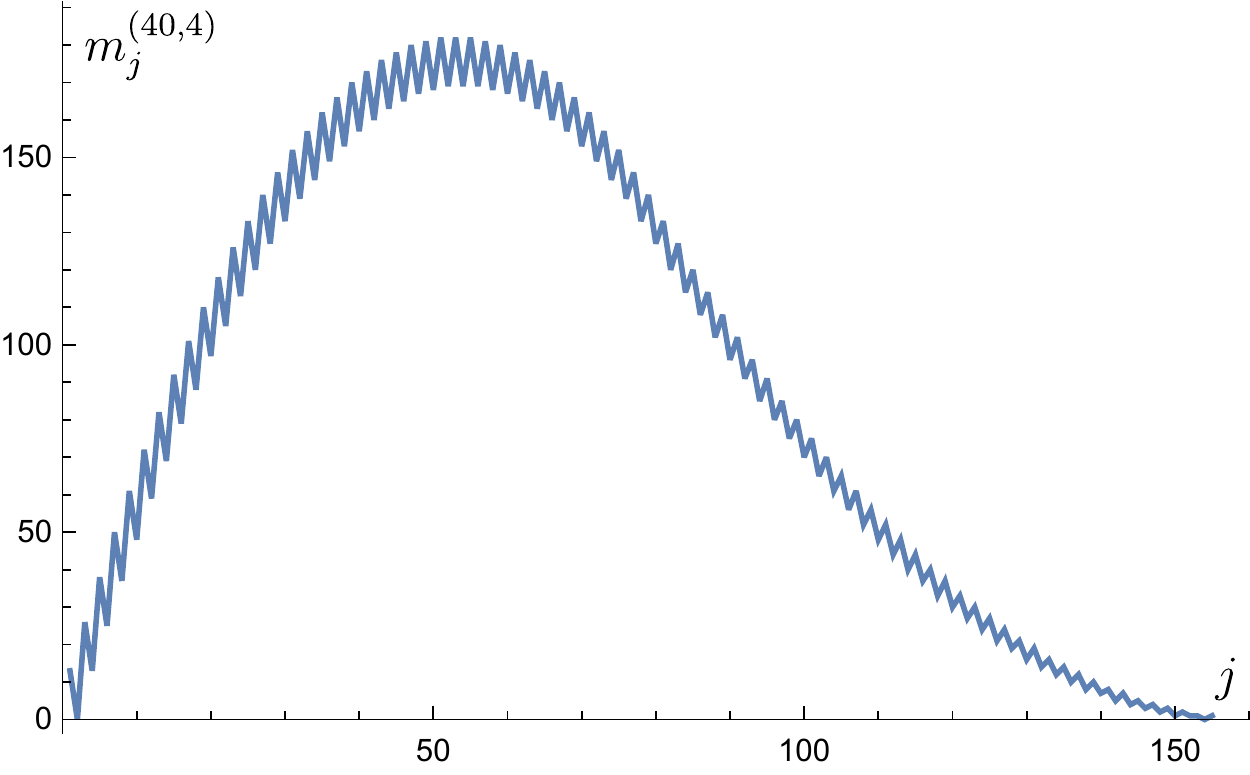}
\includegraphics[width=.48\linewidth]{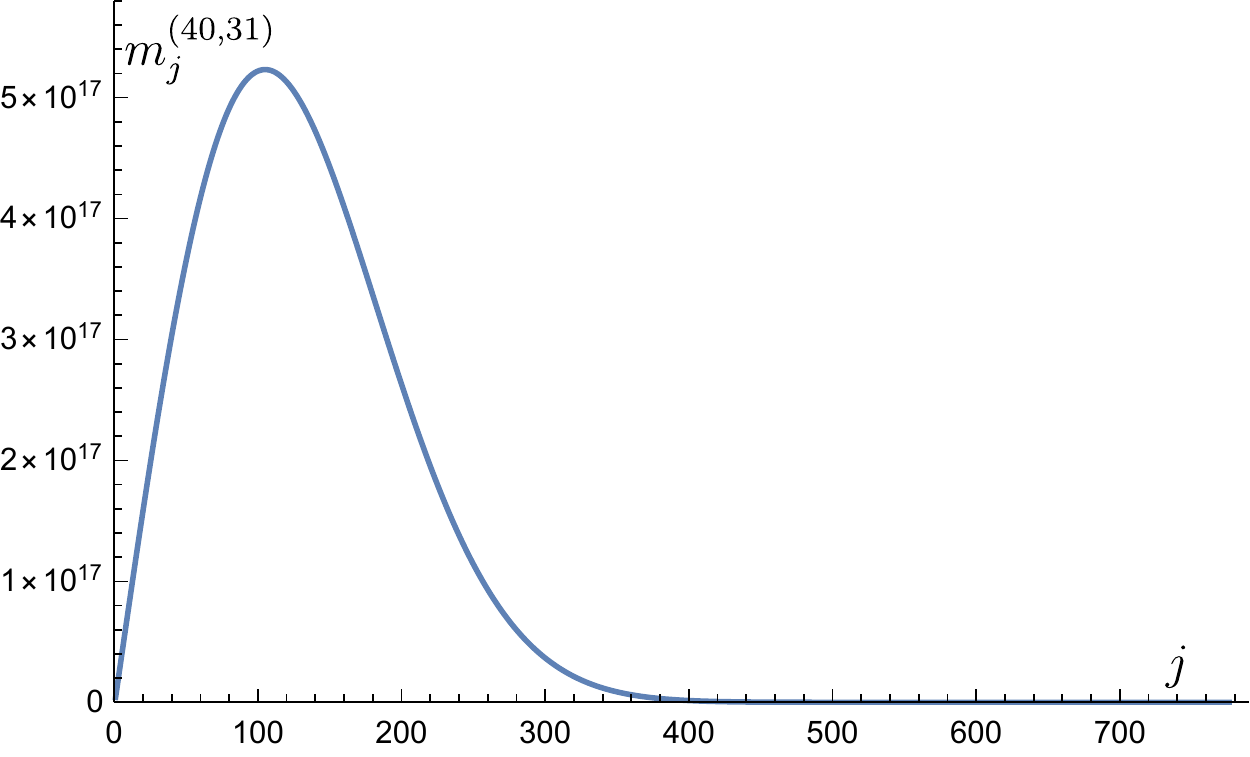}
\caption{%
Multiplicities  $m^{(40,4)}_j$ (left), $m^{(40,31)}_j$ (right), \emph{vs.} $j$,  as given by the Laurent expansion of the \rhs{} of~(\ref{genfunmskj}). In the plot on the left, the even and odd values of $j$ clearly follow, each, their own curve. Compare that plot with the one on the left in Fig.~\ref{fig:m1003j}. Pay also attention to the vertical scale in the plot on the right.
}
\label{fig:m4031j}
\end{figure}

Material related to the one presented above, concerning the multiplicities of the irreducible components of the $n$-fold tensor product of the spin-$s$ representation of $SU(2)$, can be found in~\cite{Zac:92,Cur.Kor.Zac:17}, while an enumerative combinatoric approach that rederives the above result, among many others, is undertaken in~\cite{Gya.Bar:18}. Note that the above problem of determining the irreducible components of the $k$-fold wedge product of a spin-$s$ representation is a special case of the general \emph{plethysm} problem (see p.{} 289 of~\cite{Lit:50}), which remains open to this day.
 
\begin{table}
\begin{tabular}{| c | c | c | c | c | c |c | c | c | c | c | c | c | c | c | c | c | c | c |}
\hline
& 
$\phantom{\rule[-1.7ex]{.5ex}{4.6ex}}$$\mathbf{j}$$\phantom{\rule[-.2ex]{.5ex}{2.5ex}}$
& 
$\phantom{\rule[-.2ex]{.5ex}{2.5ex}}$$\mathbf{0}$$\phantom{\rule[-.2ex]{.5ex}{2.5ex}}$ 
&
$\phantom{\rule[-.2ex]{.5ex}{2.5ex}}$$\mathbf{\frac{1}{2}}$$\phantom{\rule[-.2ex]{.5ex}{2.5ex}}$ 
&
$\phantom{\rule[-.2ex]{.5ex}{2.5ex}}$$\mathbf{1}$$\phantom{\rule[-.2ex]{.5ex}{2.5ex}}$ 
& 
$\phantom{\rule[-.2ex]{.5ex}{2.5ex}}$$\mathbf{\frac{3}{2}}$$\phantom{\rule[-.2ex]{.5ex}{2.5ex}}$
&
$\phantom{\rule[-.2ex]{.5ex}{2.5ex}}$$\mathbf{2}$$\phantom{\rule[-.2ex]{.5ex}{2.5ex}}$
&
$\phantom{\rule[-.2ex]{.5ex}{2.5ex}}$$\mathbf{\frac{5}{2}}$$\phantom{\rule[-.2ex]{.5ex}{2.5ex}}$
& 
$\phantom{\rule[-1.7ex]{.5ex}{4.6ex}}$$\mathbf{3}$$\phantom{\rule[-.2ex]{.5ex}{2.5ex}}$
&
$\phantom{\rule[-.2ex]{.5ex}{2.5ex}}$$\mathbf{\frac{7}{2}}$$\phantom{\rule[-.2ex]{.5ex}{2.5ex}}$ 
&
$\phantom{\rule[-.2ex]{.5ex}{2.5ex}}$$\mathbf{4}$$\phantom{\rule[-.2ex]{.5ex}{2.5ex}}$ 
& 
$\phantom{\rule[-.2ex]{.5ex}{2.5ex}}$$\mathbf{\frac{9}{2}}$$\phantom{\rule[-.2ex]{.5ex}{2.5ex}}$
&
$\phantom{\rule[-.2ex]{.5ex}{2.5ex}}$$\mathbf{5}$$\phantom{\rule[-.2ex]{.5ex}{2.5ex}}$
&
$\phantom{\rule[-.2ex]{.5ex}{2.5ex}}$$\mathbf{\frac{11}{2}}$$\phantom{\rule[-.2ex]{.5ex}{2.5ex}}$
&
$\phantom{\rule[-.2ex]{.5ex}{2.5ex}}$$\mathbf{6}$$\phantom{\rule[-.2ex]{.5ex}{2.5ex}}$
& 
$\phantom{\rule[-.2ex]{.5ex}{2.5ex}}$$\mathbf{\frac{13}{2}}$$\phantom{\rule[-.2ex]{.5ex}{2.5ex}}$
&
$\phantom{\rule[-.2ex]{.5ex}{2.5ex}}$$\mathbf{7}$$\phantom{\rule[-.2ex]{.5ex}{2.5ex}}$
&
$\phantom{\rule[-.2ex]{.5ex}{2.5ex}}$$\mathbf{\frac{15}{2}}$$\phantom{\rule[-.2ex]{.5ex}{2.5ex}}$
&
$\phantom{\rule[-.2ex]{.5ex}{2.5ex}}$$\mathbf{8}$$\phantom{\rule[-.2ex]{.5ex}{2.5ex}}$
\\ 
\hline
$\phantom{\rule[-1.5ex]{.5ex}{4.0ex}}$$\mathbf{s}$$\phantom{\rule[-1.7ex]{.5ex}{4.3ex}}$
& 
$\mathbf{k}$
& & & & & & & & & & & & & & & & & 
\\
\cline{1-5}
$\phantom{\rule[-1.5ex]{.5ex}{4.0ex}}$$\mathbf{1}$ & $\mathbf{2}$ & 0 &   & 1  &  &   & & & & & & & & & & & & 
\\
\cline{1-7}
$\phantom{\rule[-1.5ex]{.5ex}{4.0ex}}$$\mathbf{\frac{3}{2}}$ & $\mathbf{2}$ & 1 &   & 0  &  & 1  & & & & & & & & & & & & 
\\
\cline{1-9}
$\phantom{\rule[-1.5ex]{.5ex}{4.0ex}}$ $\mathbf{2}$ & $\mathbf{2}$ & 0  &  & 1 &  & 0 &   & 1  &  & & & & & & & & & 
\\
\cline{1-11}
$\phantom{\rule[-1.5ex]{.5ex}{4.0ex}}$ $\mathbf{\frac{5}{2}}$ & $\mathbf{2}$  & 1 &  & 0 &  & 1 &  & 0 &  & 1 &  &  &  &   &  &  &  &
\\
\cline{2-12}
$\phantom{\rule[-1.5ex]{.5ex}{4.0ex}}$& $\mathbf{3}$  &  & 0 &  & 1 &  & 1 &  & 0 &  & 1 &  &  &   &  &  &  &
\\
\cline{1-13}
$\phantom{\rule[-1.5ex]{.5ex}{4.0ex}}$ $\mathbf{3}$ & $\mathbf{2}$  & 0 &  & 1 &  & 0 &  & 1 &  & 0 &  & 1 &  &   &  &  &  &
\\
\cline{2-15}
$\phantom{\rule[-1.5ex]{.5ex}{4.0ex}}$& $\mathbf{3}$  & 1 &  & 0 &  & 1 &  & 1 &  & 1 &  & 0 &  & 1  &  &  &  &
\\
\cline{1-15}
$\phantom{\rule[-1.5ex]{.5ex}{4.0ex}}$$\mathbf{\frac{7}{2}}$ &  $\mathbf{2}$ & 1 &  & 0 &  & 1 &  & 0 &  & 1 &  & 0 &  & 1 &  &  &  &
\\
\cline{2-18}
$\phantom{\rule[-1.5ex]{.5ex}{4.0ex}}$ & $\mathbf{3}$ &  & 0 &  & 1 &  & 1 &  & 1 &  & 1 &  & 1 &    & 0 &  & 1 &
\\
\cline{2-19}
$\phantom{\rule[-1.5ex]{.5ex}{4.0ex}}$& $\mathbf{4}$ & 1 &  & 0 &  & 2 &  & 0 &  & 2 &  & 1 &  & 1 &  & 0 &  & 1
\\
\hline
\end{tabular}
\caption{Multiplicities $m^{(s,k)}_j$ of irreducible components of $D^{(s,k)}$, as given by~(\ref{multires}), and~(\ref{chis2})--(\ref{chis4}), (\ref{recchi}), or, alternatively, by~(\ref{genfunmskj}). Only values of $j$ such that $2j$ is of the same parity as $2s_\text{max}$ are considered, since for the others the multiplicities are trivially zero --- hence the empty boxes. As an example, a $k=2$ plane of spin $s=3$ decomposes into states of spin 1, 3, and 5. Note that the rightmost three entries in each row, except for the first one ($s=1$, $k=2$), are 1, 0, 1, in accordance with what was derived in the text. We have included entries up to $s=7/2$, $k=4$, because this is the lowest spin  case where a multiplicity of 2 appears, necessitating a special treatment. On the other hand, the lowest $k$ value where this happens is $k=3$, for $s=4$ (not shown in the table).}
\label{multi234}
\end{table}
\subsection{The multiconstellation of an $(s,k)$-plane}
\label{Tcoakp}
We sketched above the way to bring $D^{(s,k)}$ in block diagonal form by a change of basis in $\wedge^k \mathcal{H}$ --- we denote the unitary matrix implementing that change by $U$, while $\mathcal{D}^{(s,k)}$ will denote the block-diagonalized representation matrix (\ie, in the BD basis), with $\mathcal{D}^{(s,k)}=U D^{(s,k)} U^\dagger$. The column vector $\ket{\boldsymbol{\Psi}}$ gets transformed, accordingly, to $\ket{\boldsymbol{\Psi}}_D=U \ket{\boldsymbol{\Psi}}$, with
\begin{equation}
\label{PsiDdef}
\ket{\boldsymbol{\Psi}}_D^T
=(
\ket{\psi^{(s_\text{max})}}^T, \, 
\ket{\psi^{(s_\text{max}-2)}}^T, \,
\ldots
)
\, ,
\end{equation}
where each $\ket{\psi^{(j)}}^T$ is a row vector of $2j+1$ components ---  these irreducible multiplets are ordered in decreasing spin value. Each $\ket{\psi^{(j)}}$, defines a spin-$j$ state, and, when $j>0$, a Majorana constellation $C_j$. The full list  of these constellations, $\mathcal{C} \equiv \{C_1,C_2,\ldots, \}$, misses the information about the overall normalization and phase of each $\ket{\psi^{(j)}}$, so, to completely specify $\ket{\boldsymbol{\Psi}}_D$, we need to define a standard, normalized state $\ket{\psi_C}$, corresponding to each possible constellation $C$, by choosing arbitrarily a phase for it, and then write $\ket{\psi^{(j)}}=z_j \ket{\psi_{C_j}}$, with the complex number $z_j$ carrying now the information about the norm and overall phase of $\ket{\psi^{(j)}}$. Then the set $\{Z, \, \mathcal{C}\}$, where $Z=(z_1,z_2,\ldots)$, completely specifies 
$\ket{\boldsymbol{\Psi}}_D$. If the length of $Z$ is $2m+1$, one can view it as a spin-$m$ \emph{spectator ``state''}, and associate to it, \emph{\'a la} Majorana, a \emph{spectator constellation} $\tilde{C}$ --- then the constellations $\{\tilde{C}, \, \mathcal{C} \}$, which miss only the overall phase and normalization of $\ket{\boldsymbol{\Psi}}_D$, completely specify the $k$-plane $\Pi=[\ket{\boldsymbol{\Psi}}_D ]$. Note that, under the $SU(2)$ action on $\mathcal{H}$, the constellations in $\mathcal{C}$ rotate the way Majorana constellations do, but $\tilde{C}$ may transform in a complicated way, as the phases of the various $z_j$ (but not their moduli) may change --- we show now that, for almost all $(s,k)$-planes,  things may be arranged so that $\tilde{C}$ remains invariant under rotations.

Our treatment, at this point, will be limited by the following assumption: none of  the irreducible components 
$\ket{\psi^{(j)}}$ in the \rhs{} of~(\ref{PsiDdef}), with $j>1$,  have rotational symmetries. Regarding this, note that spin-1 states always have at least one rotational symmetry, given by a rotation by $\pi$ around the line bisecting the two stars in the Majorana constellation --- as this rotation interchanges two fermions, it imparts a phase of $\pi$ to the ket in the Hilbert space. 
Denote by 
$\tilde{\mathbb{P}}$ the corresponding projective space, with the rotationally symmetric states, excluded. Then the orbit of $\ket{\psi^{(j)}}$, under the action of $SO(3)$ is, itself, diffeomorphic to $SO(3)$ --- we call the space $\tilde{\mathcal{S}}$ of those orbits \emph{shape space}, \ie, each point in $\tilde{\mathcal{S}}$ represents an entire orbit in $\tilde{\mathbb{P}}$. Another way to describe this construction is to define an equivalence relation $\sim$ between constellations, by declaring $C'$ and $C$ to be equivalent, $C' \sim C$, iff there exists a rotation $R \in SO(3)$ such that $C'=R(C)$. That same relation can be defined in $\tilde{\mathbb{P}}$, since (non-symmetric) constellations are in 1 to 1 correspondence with states in $\tilde{\mathbb{P}}$. Then $\tilde{\mathcal{S}}=\tilde{\mathbb{P}}/\sim$, \ie, each point in shape space is an equivalence class of states in the corresponding projective space. 

Points in 
$\tilde{\mathcal{S}}$ correspond to shapes of Majorana constellations, defined, \eg, by the angles between any two stars in the constellation. Denote by $\pi$ the projection from $\tilde{\mathcal{P}}$ to $\tilde{\mathcal{S}}$, that sends each constellation $C$ to its shape $\pi(C)$. Then $\pi^{-1}(S)$ is the fiber above the shape $S$, consisting of all those constellations that share the shape $S$, and differ among themselves by a rotation. A gauge choice $\sigma$ is a map from  $ \tilde{\mathcal{S}}$ to $\tilde{\mathcal{P}}$, such that $\pi(\sigma(S))=S$, and  consists in defining a reference orientation for each shape. Given such a gauge choice, an arbitrary constellation $C$ may be defined by giving its shape 
$\pi(C)$, and the rotation $R_{\sigma,C}$, that, applied to the reference constellation (of the same shape) 
$\sigma(\pi(C))$, produces $C$, \ie, we may write 
\begin{equation}
\label{factorP}
C=(\pi(C),R_{\sigma,C})
\, ,
\quad\text{with}
\quad
R_{\sigma,C}\left( \sigma(\pi(C) \right)=C
\, .
\end{equation}
Note that, by restricting our discussion to $\tilde{\mathbb{P}}$, we guarantee that $R_{\sigma,C}$ is unique.
The algorithm for assigning a phase to a  constellation $C$, thus obtaining a state $\ket{\psi_{C}}$ in 
$\mathcal{H}$,  is then as follows: assign first, arbitrarily, a phase to the reference constellation 
$\sigma(\pi(C))$, obtaining the state 
$\ket{\psi_{\sigma,\pi(C)}}$. Then rotate this state by $R_{\sigma,C}$ to get $\ket{\psi_{C}}$, \ie, $\ket{\psi_{C}}=D^{(j)}(R_{\sigma, C}) \ket{\psi_{\sigma, \pi(C)}}$. 

Consider now a rotation $R_0$ acting on $\ket{\boldsymbol{\Psi}}_D$, \ie, $\ket{\boldsymbol{\Psi}}_D \rightarrow \ket{\boldsymbol{\Psi}'}_D = \mathcal{D}^{(s,k)}(R_0) \ket{\boldsymbol{\Psi}}_D$, inducing a transformation $\ket{\psi^{(j)}} \rightarrow \ket{{\psi'}^{(j)}}=D^{(j)}(R_0)\ket{\psi^{(j)}}$ --- at the level of constellations $C'_j=R_0(C_j)$. We have, by definition,  $\ket{{\psi'}^{(j)}}=z_j' \ket{\psi_{C'_j}}$. On the other hand, $R_{\sigma,C'_j}=R_0 \circ R_{\sigma,C_j}$, so that
\begin{align*}
{\psi'}^{(j)}
&=
D^{(j)}(R_0)\ket{\psi^{(j)}}
\\
 &=
 z_j D^{(j)}(R_0) \ket{\psi_{C_j}}
 \\
  &=
z_j D^{(j)}(R_0) D^{(j)}(R_{\sigma,C_j}) \ket{\psi_{\sigma,\pi(C_j)}}
\\
 &=
z_j  D^{(j)}(R_0 \circ R_{\sigma,C_j})\ket{\psi_{\sigma,\pi(C'_j)}}
\\
 &=
z_j D^{(j)}(R_{\sigma,C'_j}) \ket{\psi_{\sigma,\pi(C'_j)}}
\\
&=
z_j \ket{\psi_{C'_j}}
\, ,
\end{align*}
implying that $z'_j=z_j$, \ie, with the phase conventions assumed above, the spectator constellation  is invariant under rotations. 

Reference orientations for constellations are usually defined by a set of rules  that, \eg, puts one star at the north pole, a second one in the $x$-$z$ plane, with positive $x$, \etc (see, \eg,~\cite{Chr.Her:17}). Apart from the appearance of occasional ambiguities, the rules get increasingly complicated as the number of stars increases. We propose a more economic set of rules, which work, as is typical of such rules,  for almost all (but not all) constellations. Given the constellation $C$ (we drop the index $j$ for notational simplicity), with corresponding density matrix $\rho_C$, compute the spin expectation value $\vec{S}=\text{Tr}(\rho_C \mathbf{S})$, which, generically, is nonzero. Rotate $C$ to $C_1=R_1(C)$ so that $R_1(\vec{S})$ is along the positive $z$-axis, call $\rho_1=D^{(s)}(R_1) \rho_C D^{(s)}(R_1)^{-1}$ the rotated density matrix. Expand $\rho_1$ in polarization tensors~\cite{Var.Mos.Khe:88}, and identify the first non-zero component for $m\neq 0$. That component is, in general, a complex number $r e^{i\alpha}$,  rotate then $C_1$ around $z$ clockwise, by the minimal angle possible,  to make it real and positive, and call the rotated constellation $C_2=R_2(C_1)=(R_2 \circ R_1)(C)$ --- this is the reference orientation for the shape of $C$, \ie, $C_2=\sigma(\pi(C))$. A corresponding state may be defined by an arbitrary choice of phase, \eg, by taking its first nonzero component, in the $S_z$-eigenbasis, to be  real and positive. Applying to this state the unique rotation that sends $C_2$ to $C$ one gets the reference state $\ket{\psi_C}$.

A natural question that arises at this point is that of the relation between the constellations $C_j$ defined here and the principal constellation of the previous section. To elucidate this connection we need the following two results.
\begin{lemma}
\label{ipthm}
Given  $k \times n$ matrices $V$, $W$, as in~(\ref{Vdef}), and the corresponding vectors $\mathbf{V}$, $\mathbf{W}$, as in~(\ref{Plcoord}),  we have
\begin{equation}
\label{ipeqip}
\ipc{V}{W}=\braket{\mathbf{V}}{\mathbf{W}}
\, ,
\end{equation}
where  $\ipc{V}{W}$ is defined in~(\ref{ipkf}), and $\braket{\mathbf{V}}{\mathbf{W}}=\sum_{\vec{I}} \overline{V^{\vec{I}}} W^{\vec{I}}$ is the standard Hilbert space inner product.
\end{lemma}
\begin{proof}
The statement is an immediate consequence of the Cauchy-Binet formula for the expansion of a determinant (see, \eg, Sect. 2.9 of~\cite{Sha.Rem:13}).
\end{proof}
\begin{lemma}
\label{cohplue}
The irreducible components of the coherent plane $\Pi_n$ are 
\begin{equation}
\label{irrcPin}
\ket{\boldsymbol{\Pi}_n}=
\left(
\begin{array}{cccc}
\ket{n^{s_{\text{max}}}}
&
0 
&
\ldots
&
0
\end{array}
\right)
\, ,
\end{equation}
where $\ket{n^{(s_{\text{max}})}}$ is the spin-$s_{\text{max}}$ coherent state in the direction $n$.
\end{lemma}
\begin{proof}
$\Pi_n$ is obtained by $\Pi_z$ by, say, the geodesic rotation $R_{(-\sin\phi,\cos\phi,0),\theta}$ that sends $z$ to $n$. The irreducible components of $\Pi_z$ are 
\begin{equation}
\label{irrcPiz}
\ket{\boldsymbol{\Pi}_z}=
\left(
\begin{array}{cccc}
(1,0,\ldots,0)
&
0 
&
\ldots
&
0
\end{array}
\right)
\, ,
\end{equation}
where the first ket entry is the spin-$s_{\text{max}}$ coherent state along $z$, which is mapped to 
$\ket{n^{(s_{\text{max}})}}$ by the above rotation. 
\end{proof}
\begin{theorem}
\label{PrincipalC}
The principal constellation of an $(s,k)$-plane $\Pi$ coincides with the Majorana constellation of its spin-$s_{\text{max}}$ irreducible component.
\end{theorem}
\begin{proof}
Using Lemma~\ref{ipthm} the principal polynomial of $\Pi$ can be expressed in terms of $\braket{\Pi_{-n}}{\Pi}$, which, due to Theorem~\ref{cohplue}, reduces to $\braket{-n^{(s_{\text{max}})}}{\psi^{(s_{\text{max}})}}$.
\end{proof}
\begin{corollary}
\label{degreeMajoPoly_thm}
For a generic $(s,k)$-plane $\Pi$, the degree of $P_\Pi(\zeta)$ is $k\tilde{k}$.
\end{corollary}
\begin{proof}
The assertion follows immediately from the previous theorem and the fact that $s_{\text{max}}=k \tilde{k}/2$.
\end{proof}
A final remark is due regarding the case of ``degeneracy'', \ie, when the multiplicities $m_j^{(s,k)}$ are greater than 1. One then needs to choose a basis in the degenerate subspace, and let each basis element generate a spin multiplet by successive application of $S_-$. The projections of the $(s,k)$-plane $\ket{\Psi}$ onto the subspaces spanned by each of these multiplets give rise to spin-$j$ constellations, as in the non degenerate case. The salient feature here though is that the constellations thus obtained depend on the above choice of basis. The situation calls for the adoption of a particular algorithm that will single out a ``canonical'' choice of basis, much like our algorithm above for defining a standard phase for a given constellation. As the smallest example where this shows up is for a $(\frac{7}{2},4)$-plane (last line in Table~\ref{multi234}), involving  70 stars in all, we feel that, from a practical point of view, it is not necessary to spell out all the relevant details at this point. A suggestion on how to choose a canonical basis in the degenerate subspace is outlined in Example~\ref{cocb} below~\cite{Ful:97,San.Bra.Sol.Egu:17}.
\subsection{Examples}
\label{examples}
Before presenting a series of examples, we summarize, in a streamlined form,  the procedure we follow in order to derive the multiconstellation of a spin-$s$ $k$-plane $\ket{\boldsymbol{\Psi}}$. 
\begin{enumerate}
\item
Construct the BD basis and expand $\ket{\boldsymbol{\Psi}}$ in it to obtain 
$\ket{\boldsymbol{\Psi}}_D
=
\left( 
\ket{\psi^{(s_{\text{max}})}}
\ldots \ket{\psi^{(j)}} \ldots 
\right)^T$.
\item
For each irreducible component $\ket{\psi^{(j)}}$ in $\ket{\boldsymbol{\Psi}}_D$, with $j \neq 0$, determine a complex number $z^{(j)}$ as follows:
\begin{enumerate}
\item
Compute the SEV $\vec{S}^{(j)}=\bra{\psi^{(j)}} \mathbf{S} \ket{\psi^{(j)}}$ --- call $(\theta^{(j)},\phi^{(j)})$ its spherical polar coordinates (if the SEV vanishes, for any $j>0$, the procedure is not applicable).
\item
Compute the rotation matrix $R^{(j)}=\exp \left[ -i \theta^{(j)} 
\left( -\sin(\phi^{(j)}) S^{(j)}_x + \cos(\phi^{(j)}) S^{(j)}_y \right) \right]= 
\exp \left[ \frac{\theta^{(j)}}{2}
\left( e^{-i\phi^{(j)}}S_+ - e^{i\phi^{(j)}}S_- \right) \right]$.
\item
Compute $\ket{\psi^{(j)}_1}=R^{(j)} \ket{\psi^{(j)}}$, the SEV of which points along $z$.
\item
Compute $\rho^{(j)}_1=\ket{\psi^{(j)}_1}\bra{\psi^{(j)}_1}$ and expand it in polarization tensors, 
\begin{equation}
\label{rhoexp}
\rho^{(j)}_1 \rightarrow 
\left(
(\rho_{0,0}), \, 
(\rho_{1,1}, \, \rho_{1,0},\, \rho_{1,-1}),
\ldots, 
(\rho_{2j,2j}, \ldots, \rho_{2j,-2j} )
\right)
\, .
\end{equation}
Identify the first nonzero component $\rho_{\ell m} \equiv r e^{i\alpha}$, with $m \neq 0$.
\item
Compute $\ket{\psi^{(j)}_2}=e^{-i \alpha S_z/m} \ket{\psi^{(j)}_1}$ and identify its first nonzero component in the $S_z$-eigenbasis, denote the latter by $pe^{i\beta}$.
\item
Compute $z^{(j)}=\sqrt{\braket{\psi^{(j)}}{\psi^{(j)}}} e^{i\beta}$.
\item
If there is a spin-0 component $\ket{\psi^{(0)}}=(\psi^{(0)}_0)$  in $\ket{\boldsymbol{\Psi}}_D$, put $z^{(0)}=\psi^{(0)}_0$.
\end{enumerate}
\item
Determine the constellations $C_j$ for each $\ket{\psi^{(j)}}$, $j \neq 0$, as well as the spectator constellation 
$\tilde{C}$, corresponding to the ``state'' $Z=(z^{(s_{\text{max}})},\ldots,z^{(j)}, \ldots)$.
\end{enumerate}
\begin{myexample}{Irreducible component for $(1,2)$-planes}{ics1k2p}
An orthonormal basis in the Hilbert space $\mathcal{H}_{1}$ is given by the eigenvectors of $S_z$, 
$\{e_1,e_2,e_3\}
=
\{ 
\ket{1,1},
\ket{1,0},
\ket{1,-1}
\}
$. 
The associated orthonormal basis in $\mathcal{H}^{\wedge 2}$ is 
$
\{
e_{12},
e_{13},
e_{23}
\}
$,
where $e_{ij}\equiv e_i \wedge e_j$. The highest $S_z$-eigenvalue eigenvector is $e_{12}$, with eigenvalue $1+0=1$. Applying $S_-$ twice, one generates the entire spin-1 multiplet,
\begin{equation}
\label{s1multi}
\{ 
e_{(1,1)}
, \,
e_{(1,0)}
, \,
e_{(1,-1)}
\}
=
\{
e_{12}
, \,
e_{13}
, \,
e_{23}
\}
\, ,
\end{equation}
so that there is only one (spin-1)  multiplet in this case, and the matrix $U$ connecting the Pl\"ucker basis to the BD one is the identity matrix. Accordingly, $(1,2)$-planes are characterized by a single constellation of two stars, and no spectator constellation can be defined, which is as expected, as $(1,2)$-planes are the orthogonal complement of spin-1 states.

Consider  the $(1, 2)$-plane $\ket{\Sigma}=[\tilde{W}]$, with
\begin{equation}
\label{s1k2plane}
\tilde{W}=
\left(
\begin{array}{c}
\ket{\psi_1}^T
\\
\ket{\psi_2}^T
\end{array}
\right)
=
\left(
\begin{array}{ccc}
1 & 0 & i
\\
0 & 1 & 1-i
\end{array}
\right)
\, .
\end{equation}
The Majorana constellations of the two kets $\ket{\psi_1}$, $\ket{\psi_2}$, spanning $\ket{\Sigma}$, are
\begin{equation}
\label{constpsi12}
\left\{ n_{11},n_{12}\right\}
=
\left\{ 
\big(
\frac{1}{\sqrt{2}},-\frac{1}{\sqrt{2}},0 
\big)
\, , \, 
\big(
-\frac{1}{\sqrt{2}},\frac{1}{\sqrt{2}},0 
\big) 
\right\}
\, ,
\qquad
\left\{ 
n_{21},n_{22}
\right\}=
\left\{ 
\big(
\frac{1}{\sqrt{2}},-\frac{1}{\sqrt{2}},0 
\big)
\, , \, 
\big(
0,0,-1 
\big) 
\right\}
\, ,
\end{equation}
respectively. The (unnormalized) Pl\"ucker (and BD) components of $\ket{\Sigma}$ are $\ket{\Sigma}=\ket{\Sigma}_D=(1,1-i,-i)$, with constellation
\begin{equation}
\label{Sigmaconst}
\left\{ n_A,n_B \right\}=
\left\{ 
\big(
\frac{1}{\sqrt{2}},-\frac{1}{\sqrt{2}},0 
\big)
\, , \,
\big(
\frac{1}{\sqrt{2}},-\frac{1}{\sqrt{2}},0 
\big) 
\right\}
\, ,
\end{equation}
\ie, $\ket{\Sigma}$ is a $(1,2)$-coherent plane. We note that the functional relationship $n_{A,B}(n_{ij})$, even in this, simplest of cases, is surprisingly complicated.
\end{myexample}
\begin{myexample}{Irreducible components for $(\frac{3}{2},2)$-planes}{ics322p}
An orthonormal basis in the Hilbert space $\mathcal{H}_{\frac{3}{2}}$ is given by the eigenvectors of $S_z$, 
$\{e_1,e_2,e_3, e_4\}
=
\{ 
\ket{\frac{3}{2},\frac{3}{2}},
\ket{\frac{3}{2},\frac{1}{2}},
\ket{\frac{3}{2},-\frac{1}{2}},
\ket{\frac{3}{2},-\frac{3}{2}}
$. 
The associated orthonormal basis in $\mathcal{H}^{\wedge 2}$ is 
$
\{
e_{12},
e_{13},
e_{14},
e_{23},
e_{24} ,
e_{34}
\}
$,
where $e_{ij}\equiv e_i \wedge e_j$. The highest $S_z$-eigenvalue eigenvector is $e_{12}$, with eigenvalue $\frac{3}{2}+\frac{1}{2}=2$. Applying $S_-$ four times, one generates the entire spin-2 multiplet,
\begin{equation}
\label{s2multi}
\{ 
e_{(2,2)}
, \,
e_{(2,1)}
, \,
e_{(2,0)}
, \,
e_{(2,-1)}
, \,
e_{(2,-2)}
\}
=
\{
e_{12}
, \,
e_{13}
, \,
\frac{1}{\sqrt{2}} (e_{14}+ e_{23})
, \,
e_{24}
, \,
e_{34}
\}
\, .
\end{equation}
The $S_z$-eigenvalue 0 is doubly degenerate, the state orthogonal to $e_{(2,0)}$ is the spin-0 state 
$e_{(0,0)}=(e_{14}-e_{23})/\sqrt{2}$. The matrix $U$ effecting the change between the two bases, 
$\ket{\boldsymbol{\Psi}}_D=U \ket{\boldsymbol{\Psi}}$, is
\begin{equation}
\label{Us1k2m}
U=
\left(
\begin{array}{cccccc}
1 & 0 & 0 & 0 & 0 & 0
\\
0 & 1 & 0 & 0 & 0 & 0 
\\
0 & 0 & \frac{1}{\sqrt{2}} & \frac{1}{\sqrt{2}} & 0 & 0
\\
0 & 0 & 0 & 0 & 1 & 0
\\
0 & 0 & 0 & 0 & 0 & 1
\\
0 & 0 & \frac{1}{\sqrt{2}} & -\frac{1}{\sqrt{2}} & 0 & 0
\end{array}
\right)
\, .
\end{equation}

Consider now  the two 2-planes $\tilde{W}_1$, $\tilde{W}_2$, encountered in Example~\ref{Ts322planes} (see equation~(\ref{twosols32k2})), which shared the same principal polynomial, $\zeta^4-1$, and, hence, principal constellation (a square on the equator). The rows of $\tilde{W}_1$, after normalization, are $(e_1+i e_3)/\sqrt{2}$, $(e_2+i e_4)/\sqrt{2}$, so that the 2-plane $\tilde{W}_1$ represents is
\begin{align*}
\ket{\Sigma_1}=[\tilde{W}_1]
&=
\frac{1}{2} (e_1+i e_3) \wedge (e_2+i e_4)
\\
 &=
\frac{1}{2} (e_{12}+i e_{14} -i e_{23} -e_{34})
\\
& \rightarrow 
\frac{1}{2}
\left(
\begin{array}{cccccc}
1 & 0 & i & -i & 0 & -1
\end{array}
\right)^T
\, ,
\end{align*}
and, similarly, $\ket{ \Sigma_2}=\left(
\begin{array}{cccccc}
1 & -i & i & 0 & -1
\end{array}
\right)^T/2$. Left-multiplying by $U$ we find their irreducible components,
\begin{equation}
\label{Sigma12irred}
\ket{\Sigma_1}_D=\frac{1}{2} 
\left( 
\begin{array}{cccccc} 
\big(1 & 0 & 0 & 0 & -1 \big) & \big( i\sqrt{2} \big)
\end{array}
\right)^T
\, ,
\qquad
\ket{\Sigma_2}_D=\frac{1}{2} 
\left( 
\begin{array}{cccccc} 
\big( 1 & 0 & 0 & 0 & -1 \big) & \big( -i\sqrt{2} \big)
\end{array}
\right)^T
\, ,
\end{equation}
where we used extra parentheses to visually define the spin-2 quintet and the  spin-0 singlet. Note that the spin-2 component, which gives rise to the principal constellation, is identical in the two planes, which, however, are distinguished by their differing spin-0 components. There are two reasons why our procedure for determining the spectator constellation is not applicable in this case: the principal constellation has nontrivial rotation symmetries (\eg, a rotation around $z$ by $\pi/2$) and the SEV of the spin-2 component vanishes. This is a good example of why our requirement of non-symmetric constellations is necessary for the definition of the spectator constellation: under the above mentioned symmetry rotation, the principal constellation of both planes goes back to itself, but the corresponding spin-2 state picks up a sign, resulting in the rotation exchanging the two planes --- this would contradict our result that the spectator constellation is invariant under rotations.
\end{myexample}
\begin{myexample}{Multiconstellation  for a $(2,2)$-plane}{Ms22p}
Proceeding as in the previous example, we find for the matrix $U$ transforming from the Pl\"ucker to the BD basis,
\begin{equation}
\label{Us2k2}
U
=
\left(
\begin{array}{cccccccccc}
 1 & 0 & 0 & 0 & 0 & 0 & 0 & 0 & 0 & 0
   \\
 0 & 1 & 0 & 0 & 0 & 0 & 0 & 0 & 0 & 0
   \\
 0 & 0 & \sqrt{\frac{3}{5}} & 0 &
   \sqrt{\frac{2}{5}} & 0 & 0 & 0 & 0
   & 0 \\
 0 & 0 & 0 & \frac{1}{\sqrt{5}} & 0 &
   \frac{2}{\sqrt{5}} & 0 & 0 & 0 & 0
   \\
 0 & 0 & 0 & 0 & 0 & 0 &
   \sqrt{\frac{3}{5}} &
   \sqrt{\frac{2}{5}} & 0 & 0 \\
 0 & 0 & 0 & 0 & 0 & 0 & 0 & 0 & 1 & 0
   \\
 0 & 0 & 0 & 0 & 0 & 0 & 0 & 0 & 0 & 1
   \\
 0 & 0 & \sqrt{\frac{2}{5}} & 0 &
   -\sqrt{\frac{3}{5}} & 0 & 0 & 0 & 0
   & 0 \\
 0 & 0 & 0 & \frac{2}{\sqrt{5}} & 0 &
   -\frac{1}{\sqrt{5}} & 0 & 0 & 0 & 0
   \\
 0 & 0 & 0 & 0 & 0 & 0 &
   \sqrt{\frac{2}{5}} &
   -\sqrt{\frac{3}{5}} & 0 & 0 \\
\end{array}
\right)
\, .
\end{equation}
Take, as example, the 2-plane $\ket{\boldsymbol{\Psi}}=v \wedge w$, where 
\begin{equation}
\label{vwdef}
v=
\left(
\begin{array}{ccccc}
1 & 0 & 1 & 0 & 0
\end{array}
\right)^T
\, ,
\qquad
w
=
\left(
\begin{array}{ccccc}
1 & 0 & 0 & 0 & 1
\end{array}
\right)^T
\, .
\end{equation}
Its normalized Pl\"ucker components are
\begin{equation}
\label{normPcomp}
\ket{\hat{\boldsymbol{\Psi}}}
=
\frac{1}{2}
\left(
\begin{array}{cccccccccc}
1 & 0 & 0 & 1 & -1 & 0 & 0 & 0 & 1 & 0
\end{array}
\right)^T
\, .
\end{equation}
Left multiplication by the above $U$ gives 
\begin{equation}
\label{PsiBDcomps}
\ket{\hat{\boldsymbol{\Psi}}}_D
=
\left(
\begin{array}{cc}
\ket{\psi^{(3)}}^T & \ket{\psi^{(1)}}^T
\end{array}
\right)^T
=
\frac{1}{\sqrt{20}}
\left(
\begin{array}{cccccccccc}
\big(
\sqrt{5} & 0 & -\sqrt{2} & 1 &  0 & \sqrt{5} & 0 
\big)
& 
\big(
\sqrt{3}  & 2 & 0
\big)
\end{array}
\right)^T
\, ,
\end{equation}
where the extra parentheses define visually the spin-3 and spin-1 multiplets. 
Each of $\ket{\psi^{(3)}}$, $\ket{\psi^{(1)}}$ has its own Majorana constellation. But the two states are not normalized to unity, and their constellations also miss the information about their phase. Both pieces of information are captured in the spectator spin-1/2 state $Z=(z_3,z_1)$, which we now determine. 

The SEV for $\ket{\psi^{(3)}}$  is $\vec{S}^{(3)}=(-\sqrt{\frac{3}{50}},0,\frac{7}{20})$, with polar coordinates 
$(\theta^{(3)},\phi^{(3)})=(\arctan \frac{20\sqrt{3}}{7\sqrt{50}},\pi)$. We compute the rotated state $\ket{\psi^{(3)}_1}$, and expand the corresponding density matrix in polarization tensors  to find
\begin{equation}
\rho^{(3)}_1 \rightarrow \left(
\left(
\frac{13}{20\sqrt{7}}
\right),
\left(
0,\frac{\sqrt{\frac{73}{7}}}{40}, 0 
\right),
\left(
\frac{31}{730\sqrt{14}}, -\frac{29}{146\sqrt{21}},\frac{241\sqrt{\frac{3}{7}}}{2920},\frac{29}{146 \sqrt{21}},
\frac{31}{730\sqrt{14}}
\right),
\ldots
\right)
\, .
\end{equation}
Note that the spin-1 component in this expansion is of the form $\{ 0,\lambda,0\}$, with $\lambda>0$, as a result of the SEV of $\rho^{(3)}_1$ being along $z$.
The first nonzero component, with $m\neq 0$, is the 22-component, which is already real and positive, so the second rotation, around the $z$-axis, is the identity, and $\ket{\psi^{(3)}_2}=\ket{\psi^{(3)}_1}=\left( 0.258,0.581,\ldots\right)$, the last equality giving the components of $\ket{\psi^{(3)}_2}$ in the $S_z$ eigenbasis.
Because the first nonzero component is real and positive, we get $z^{(3)}=\sqrt{\braket{\psi^{(3)}}{\psi^{(3)}}}=\sqrt{13/20}$.

Proceeding analogously we find $z^{(1)}=i\sqrt{7/20}$, so that the spectator ``state'' is 
$Z=\left(\sqrt{13/20},i\sqrt{7/20}\right)$. 
A plot of the corresponding constellations appears in figure~\ref{fig:psi31StarsPlot}.
\begin{figure}[h]
\hspace{2ex}
\includegraphics[width=.3\linewidth]{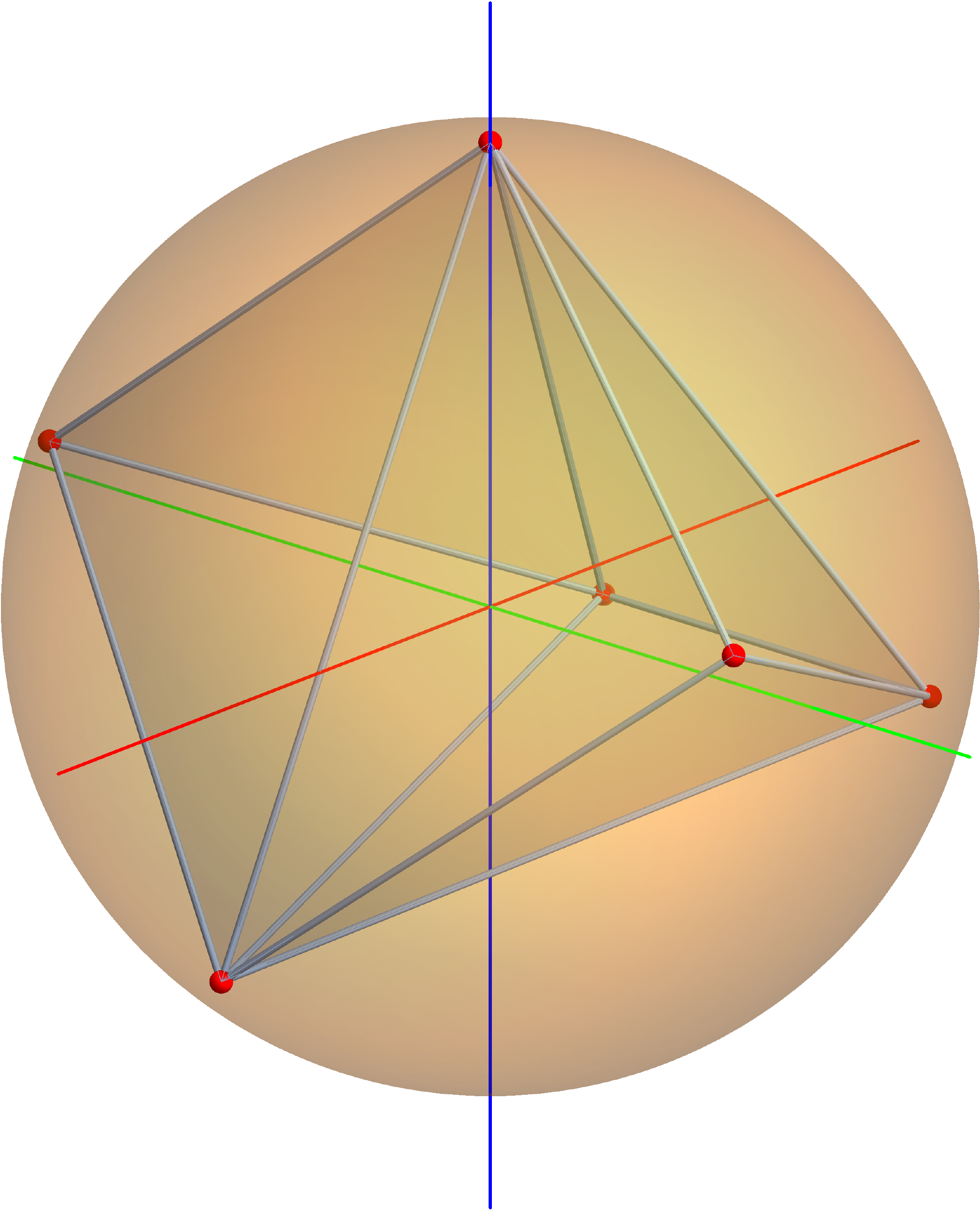}
\hspace{2ex}
\includegraphics[width=.3\linewidth]{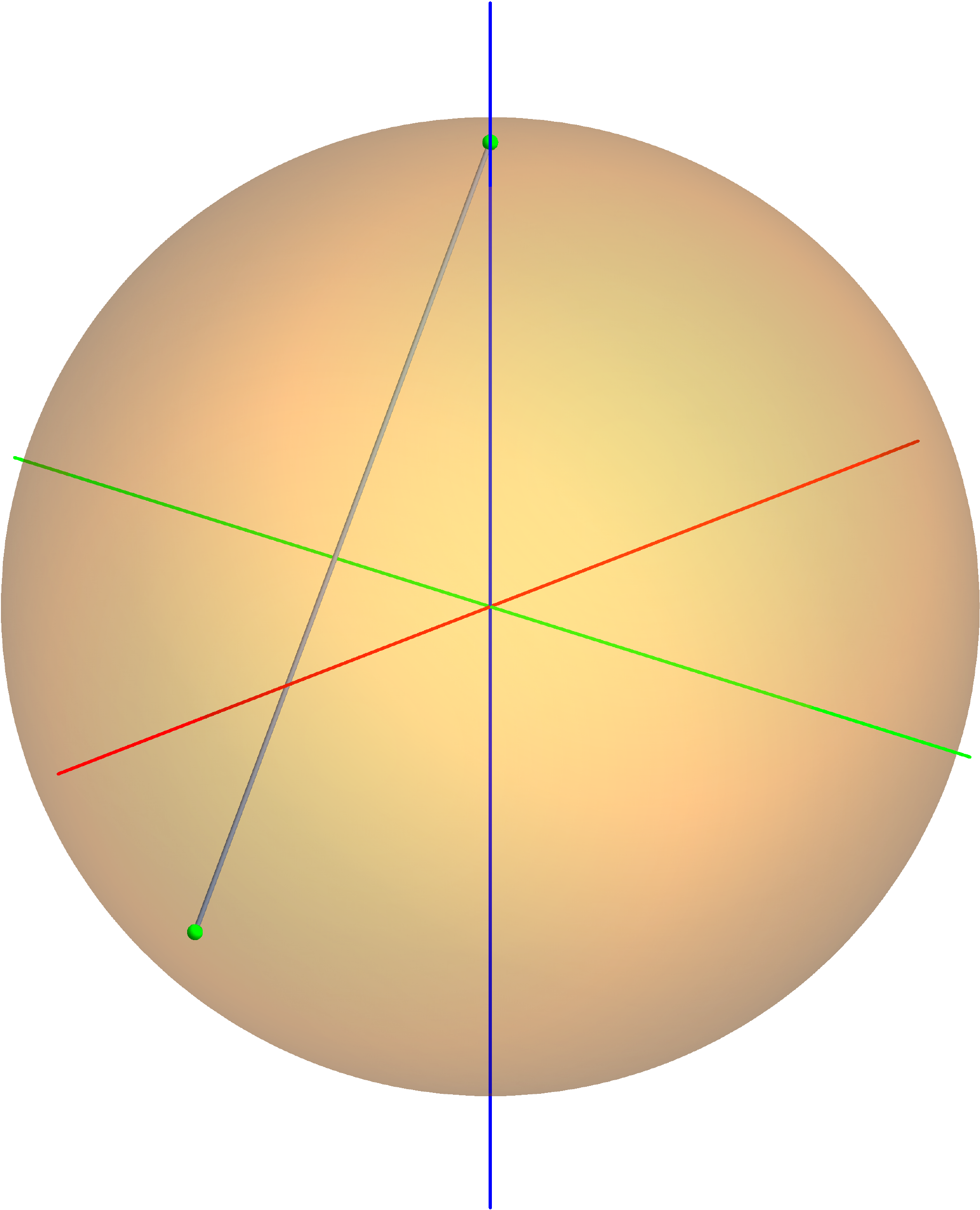}
\hspace{2ex}
\includegraphics[width=.3\linewidth]{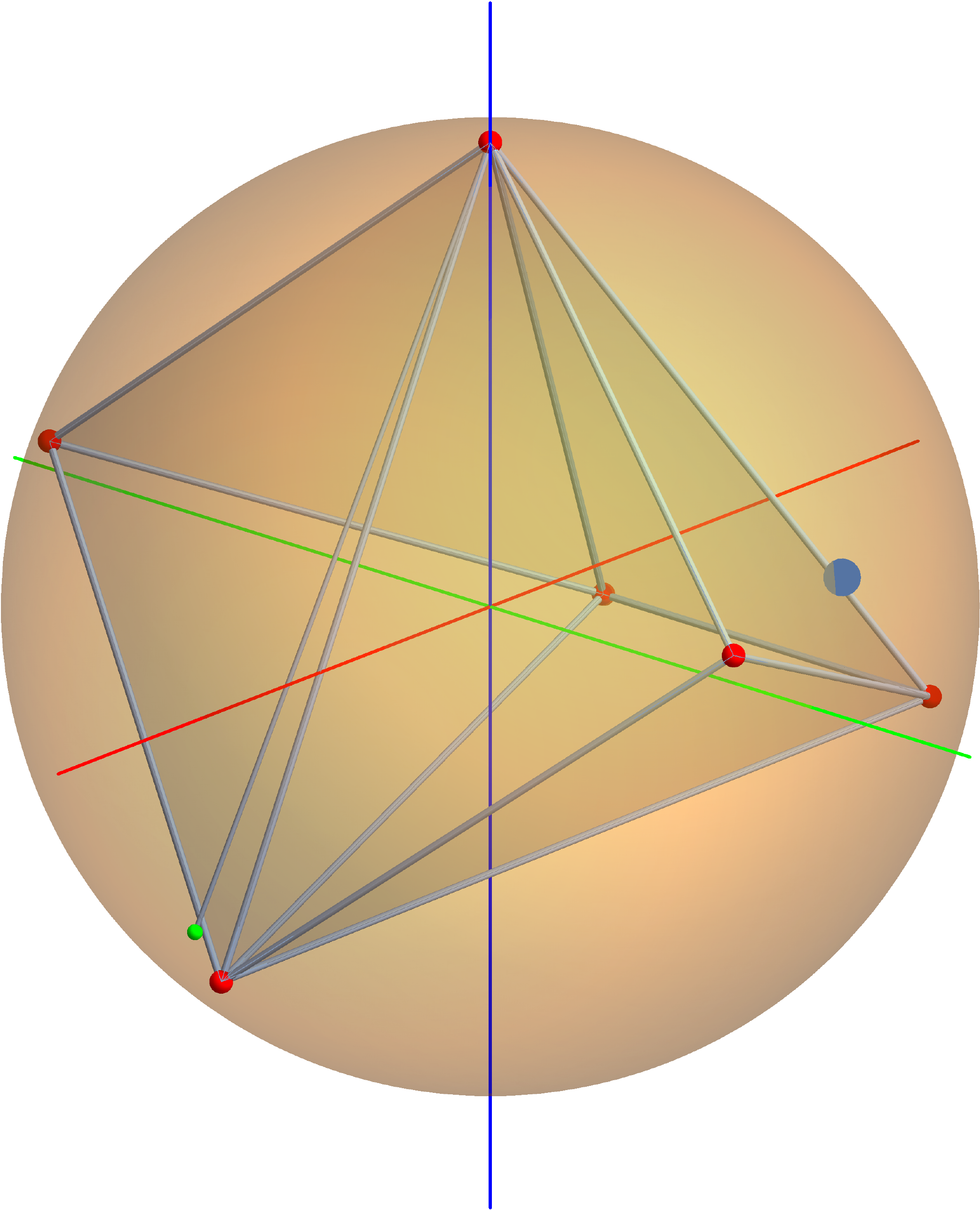}
\hspace{2ex}
\caption{%
Constellations $C_3$ (left), $C_1$ (middle), and both superimposed, including the spectator $\tilde{C}$ (blue dot) (right) of the $(2,2)$-plane in (\ref{vwdef}) --- note that $C_3$ and $C_1$  both have a star at the north pole.
}
\label{fig:psi31StarsPlot}
\end{figure}
\end{myexample}
\begin{myexample}{Choice of canonical basis in the $(\frac{7}{2},4)$ case}{cocb}
We start with the top plane $\ket{\frac{7}{2}}\wedge \ket{\frac{5}{2}}\wedge \ket{\frac{3}{2}} \wedge \ket{\frac{1}{2}}$, with $S_z=8$, and generate the entire $s=8$ multiplet by repeated application of $S_-$. At $S_z=6$, a second state appears (apart from the one belonging to the above multiplet), that generates, similarly, an $s=6$ multiplet. Then, at $S_z=4$, two new states appear (apart from the ones belonging to the previous two multiplets),
\begin{align}
\label{Pi1}
\ket{\boldsymbol{\Psi}_1} 
&\sim 
7\sqrt{3} \ket{\frac{7}{2}} \wedge \ket{\frac{5}{2}} \wedge \ket{-\frac{1}{2}} \wedge \ket{-\frac{3}{2}}
-14 \ket{\frac{7}{2}} \wedge \ket{\frac{3}{2}} \wedge \ket{\frac{1}{2}} \wedge \ket{-\frac{3}{2}}
+2\sqrt{105} \ket{\frac{5}{2}} \wedge \ket{\frac{3}{2}} \wedge \ket{\frac{1}{2}} \wedge \ket{-\frac{1}{2}}
\, ,
\\
\ket{\boldsymbol{\Psi}_2} 
&\sim 
2\sqrt{105} \ket{\frac{7}{2}} \wedge \ket{\frac{5}{2}} \wedge \ket{\frac{3}{2}} \wedge \ket{-\frac{7}{2}}
-14 \ket{\frac{7}{2}} \wedge \ket{\frac{5}{2}} \wedge \ket{\frac{1}{2}} \wedge \ket{-\frac{5}{2}}
+7\sqrt{3} \ket{\frac{7}{2}} \wedge \ket{\frac{5}{2}} \wedge \ket{-\frac{1}{2}} \wedge \ket{-\frac{3}{2}}
\, ,
\end{align}
which are degenerate in their expectation value of  $S_z=\sum_{r=1}^4 S_z^{[r]} \equiv Q^{(1)}$, where $S_z^{[r]}$ is the $S_z$ operator in the $r$-th wedge factor. One may similarly define the operator $Q^{(2)}=\sum_{r=1}^4 (S_z^{[r]})^2$, and distinguish the two states above according to their $Q^{(2)}$ expectation value. To this end, we consider the linear combination 
$\ket{\boldsymbol{\Psi}}=\alpha \ket{\boldsymbol{\Psi}_1}+\beta \ket{\boldsymbol{\Psi}_2}$, normalized to 1, and maximize $\bra{\boldsymbol{\Psi}}Q^{(2)} \ket{\boldsymbol{\Psi}}$ to find $\beta=(-109+4\sqrt{715})/21 \alpha$, which defines the first vector in the canonical basis we are after, while the second one is defined by orthogonality. When the degeneracy is greater than 2, additional, higher order, operators $Q^{(n)}$ may be used to lift it.
\end{myexample}

\noindent 

\section{Epilogue}
\label{Epilogue}
We have presented a generalization of Majorana's stellar representation of spin quantum states to the case of $(s,k)$-planes through the origin  in Hilbert space. Given such a plane, we first constructed an associated  Majorana-like  principal constellation, that rotates in physical space as the plane is rotated in Hilbert space. We then showed how to augment  this construction to a family of constellations, which, unlike the principal constellation, uniquely characterizes the plane. 

We mention here briefly possible applications of the above results. As alluded to already in the introduction, being able to visualize an $(s,k)$-plane simplifies the task of identifying its rotational symmetries. It is self-evident that the rotational symmetry group of any $(s,k)$-plane is a subgroup of the intersection of the symmetry groups of each of its constellations, since the invariance of the latter under a rotation is a necessary condition for the invariance of the plane. The condition, however, is not sufficient, because a constellation $C_i$ coming back to itself after a rotation implies that the corresponding state $\ket{\psi^{(i)}}$ might acquire a phase, and if the phases of the various $\ket{\psi}$'s that appear in~(\ref{PsiDdef}) are not equal,  the plane will not be invariant under the rotation --- we saw this happening in Example~\ref{ics322p}. On the other hand, if the above intersection of symmetry groups is trivial, the plane has no rotational symmetries, as is the case, in particular, if any of the $C_i$ has no such symmetries --- it is hard to see how to reach such a conclusion without the aid of our construction. Note that the converse problem is not trivial: our discussion above \emph{does not} clarify how to construct an $(s,k)$-plane with given rotational symmetries. It is true that one may choose freely the principal constellation, in particular endowing it with any desired symmetry, but the secondary constellations that complete the multiconstellation cannot be fixed at will --- rather, they can only take a discrete set of values, the determination of which, given the principal constellation, is rather non-trivial. We defer the elucidation of these matters to a future publication, currently in progress.

Another instance where our results might prove useful is the visualization of multifermionic spin states. Indeed, it is clear that the Grassmannian $\text{Gr}(k,n)$ may be thought of as the subspace of wedge-factorizable antisymmetric states of $k$ spin-$s$ particles, in which case the ambient Pl\"ucker space is just the full antisymmetric state space. There is nothing in our construction of multiconstellations that limits it to wedge-factorizable states though, so it can be used as well to visualize an arbitrary antisymmetric state, codifying, in particular, its rotational symmetries, as outlined above. Again, it is difficult to see how to efficiently unveil this information by other means. Note that totally antisymmetric $k$-partite states have long held a prominent role in atomic and molecular physics, where, when $k=\tilde{N}$, they are known as Slater-determinant states. These have also proved useful in quantum information processing~\cite{Saw.Huc.Kus:11,Sch.Cir.Kus.Lew.Los:01}, in which context they can be generated iteratively by a sequence of generalized XOR-gates and discrete Fourier transforms, and have also applications in, \eg, quantum 
cryptography~\cite{Jex.Alb.Bar.Del:03}.

There are various directions along which the above ideas may be further developed. A question we consider most pressing is the clarification of the physical meaning of the principal constellation. The analogous question for the Majorana constellation of a spin-$s$ state $\ket{\psi}$ has a concise, and conceptually appealing answer involving the $2s$ spin-1/2 particles whose symmetrization gives rise to $\ket{\psi}$. We feel that a similarly simple and appealing answer ought  to exist for the principal constellation. Another direction worth  exploring is the significance of coincident stars in a constellation. Such degenerate constellations clearly represent singular points in the Grassmannian, already in the original case of Majorana, and their mathematical description involves secant and  tangent varieties (see, \eg,~\cite{Zak:93,Hey:08,Hol.Luq.Thi:12,Chr.Guz.Ser:18}) --- we hope we will soon be able to report our progress on these matters. On the applications front, our first priority would be to develop possible ramifications for the program of holonomic quantum computation~\cite{Zan.Ras:99}. The Wilzcek-Zee effect, upon which this entire endeavor is based, considers a $k$-dimensional degenerate subspace of the Hilbert space that undergoes cyclic evolution, tracing a closed curve in the corresponding Grassmannian. The practical problem one faces at the outset with this requirement  is identifying the closure of the curve, as a particular  basis in the plane may not return to itself, even when the plane it spans does. Clearly, representing the plane by its multiconstellation solves this problem, and further simplifies it in the case the  time evolution of the plane in question corresponds to a sequence of rotations, as the latter may be applied directly to the multiconstellation. 
\section*{Acknowledgements}
\label{Ack}
The authors would like to acknowledge partial financial support from UNAM-DGAPA-PAPIIT project IG100316. ESE would also like to acknowledge financial support from the T@T fellowship of the Univesrity of T\"ubingen.

\end{document}